\documentclass[12pt,a4paper]{article}
\usepackage[a4paper,left=1.25in,right=1.25in,top=1.25in,bottom=1.25in]{geometry}
\usepackage{amssymb,amsmath,amsthm,amsfonts}
\usepackage{graphicx,color}
\usepackage{wasysym}
\usepackage{enumitem}
\usepackage{tikz}
\usepackage{subfigure}
\usepackage{booktabs}
\usepackage{dcolumn}
\usetikzlibrary{plotmarks}
\usepackage[footnotesize]{caption}
\usepackage{dcolumn}
\usepackage{rotating}
\usepackage{eurosym}
\usepackage{natbib}
\usepackage{hyperref}
\hypersetup{pdftitle={Best-response dynamics in directed network games},pdfauthor={P\'eter Bayer, Gy\"orgy Kozics, N\'ora Gabriella Sz\"oke},pdfkeywords={Networks, externalities, local public goods, potential games, non-reciprocal relations}}
\usepackage{hyphenat}
\usepackage[normalem]{ulem} 

\newcolumntype{d}[1]{D{.}{.}{#1}}

\usepackage[latin2]{inputenc}

\newtheorem{Theorem}{theorem}[section]
\newtheorem{theorem}[Theorem]{Theorem}

\newtheorem{claim}[Theorem]{Claim}

\newtheorem{corollary}[Theorem]{Corollary}

\newtheorem{lemma}[Theorem]{Lemma}

\newtheorem{proposition}[Theorem]{Proposition}
\theoremstyle{definition}
\newtheorem{definition}[Theorem]{Definition}
\newtheorem{assumption}[Theorem]{Assumption}
\newtheorem{example}[Theorem]{Example}
\theoremstyle{remark}
\newtheorem{remark}[Theorem]{Remark}

\DeclareMathOperator{\sgn}{sgn}
\newcommand{\sqrtabs}[1]{\sqrt{|{#1}|}}

\usepackage{tikz}

\DeclareMathOperator*{\argmax}{argmax\,} 
\usetikzlibrary{shapes}

\newcommand*\xh{\overline{x}}

\newcommand*\bt{\tilde{b}}
\newcommand*\uw{\underline{w}}
\newcommand*\ow{\overline{w}}
\newcommand*\tw{\tilde{w}}
\newcommand*\ot{\overline{t}}
\newcommand*\ut{\underline{t}}
\newcommand*\eps{\varepsilon}
\newcommand*\N{\mathbb{N}}
\newcommand*\R{\mathbb{R}}
\newcommand*\C{\mathbb{C}}

\begin{document}
\baselineskip=1.35\baselineskip

\title{Best-response dynamics in directed network games\thanks{We thank Yann Bramoull\'e, P\'eter Cs\'oka, P. Jean-Jacques Herings, Botond K\H{o}szegi, M\'anuel L\'aszl\'o M\'ag\'o, Ronald Peeters, Mikl\'os Pint\'er, \'Ad\'am Szeidl, Frank Thuijsman, Yannick Viossat, and Mark Voorneveld for valuable feedback and suggestions.}}
\author{P\'eter Bayer\thanks{Corresponding author. Universit\'e Grenoble Alpes, GAEL, 1241 Rue des R\'esidences, 38400 Saint Martin d'H\`eres, France. E-mail: peter.bayer@univ-grenoble-alpes.fr.}
\and Gy\"orgy Kozics\thanks{Central European University, Department of Economics and Business, N\'ador utca 9, 1051 Budapest, Hungary. E-mail: kozics\_gyorgy@phd.ceu.edu.}
\and N\'ora Gabriella Sz\H{o}ke\thanks{Universit\'e Grenoble Alpes, Institut Fourier, CS 40700, 38058 Grenoble cedex 09, France. E-mail: nora.gabriella.szoke@gmail.com. Research supported by the Swiss National Science Foundation, Early Postdoc.Mobility fellowship no.~P2ELP2\_184531.}}
\date{\today}
\maketitle

\begin{abstract}
\baselineskip=1.35\baselineskip
\noindent We study public goods games played on networks with possibly non-recip\-rocal relationships between players. Examples for this type of interactions include one-sided relationships, mutual but unequal relationships, and parasitism. It is well known that many simple learning processes converge to a Nash equilibrium if interactions are reciprocal, but this is not true in general for directed networks. However, by a simple tool of rescaling the strategy space, we generalize the convergence result for a class of directed networks and show that it is characterized by transitive weight matrices. Additionally, we show convergence in a second class of networks; those rescalable into networks with weak externalities. We characterize the latter class by the spectral properties of the absolute value of the network's weight matrix and show that it includes all directed acyclic networks.
\\*[0.5\baselineskip]
\\*[0.5\baselineskip]
\end{abstract}

\newpage

\section{Introduction}
All social and economic networks feature relationships which cannot be described as a partnership of equals. There are relationships between pairs of agents in which only one party is interested and the other is not. Some relationships are beneficial for one party but harmful for the other. Even when both parties benefit or both are harmed, the extent by which they are affected by their counterpart's decisions is not necessarily equal. Nevertheless, most scientific works, both theoretical and applied, are using models in which the reciprocity of interactions is a fundamental property. These frameworks, simple graphs and weighted networks, have received a lot of recent interest from economic theorists due to them providing a highly accurate, rich, and efficient description of real-life interactions. However for reasons relating to either convenience or convention, non-reciprocal relations, and their relevance to economic theory, are relatively unexplored in network literature.	

In particular, most models featuring externalities in networks such as \cite{Ballester2006} and \cite{BramoulleKranton2007} assume  reciprocal interactions. These highly influential theoretical papers opened the way for a number of applications, such as R\&D expenditure between interlinked firms \citep{Konig2019}, peer effects \citep{Blume2010}, defense expenditures \citep{SandlerHartley1995,SandlerHartley2007}, and crime \citep{Ballester2010}. Most of the applied literature continues to assume reciprocity of relations and performs equilibrium analysis. However, as it will be apparent from our results, the basis behind using the Nash equilibrium as a prediction in networks with non-reciprocal interactions is much weaker than with only reciprocal ones. Thus, predictions made by models using simple graphs or weighted networks are only justified for applications where the underlying interaction network is shown to consist only of reciprocal relations.

In this paper we extend the theoretical literature of learning in networks to include non-reciprocal relations by the use of directed networks. By doing so we intend to fill a gap in the theory literature and provide a jumping off point for a stronger connection between subsequent applied literature and real-life interaction networks. Our setting generalizes games played on weighted networks: instead of one weight, each link is defined by two distinct weights, one for each direction. The three important types of non-reciprocal relations we highlight are (1) one-directional links with one of the weights of the link being zero and the other being non-zero, e.g., an upstream city along a river affects a neighboring downstream city by polluting the river but not vice versa, (2) parasitic links with one weight being positive and the other being negative, e.g., a criminal organization engaging in the extortion of a small business gains benefits from the interaction but it comes with losses for the business, and (3) mutual but unequal benefit or harm, with both weights being positive or both being negative but they are not equal, e.g., a monopolistic seller and a competitive buyer both benefit from their economic relationship but the gains from the interaction tend to be greater for the seller than for the buyer.

One of the focuses of theoretical literature of public goods and networks is the convergence of learning processes to the Nash equilibrium. This is to provide a behavioral and an evolutionary motivation to use equilibrium as a prediction in applied settings. Under symmetric weight matrices a number of powerful convergence results are known: Stability of Nash equilibria with respect to the continuous best-response dynamic has been established by \cite{Bramoulle2014}. Convergence of the continuous best-response dynamic to some Nash equilibrium in all such games has been shown by \cite{BervoetsFaure2019}. \cite{BHPT2019} shows convergence of a class of one-sided learning processes. \cite{Bervoets2018} constructs a convergent learning process not requiring the sophistication of the best response. \cite{BHP2019} studies the impact of a farsighted agent on the process in a population of myopic players. Together, these results allow for an interpretation of the game's Nash equilibria as the results of a sequence of improvements made separately by the players.

All of the above papers assume reciprocal network interactions. Thus, they make use of more general results of games of reciprocal interactions \citep{Dubey2006,Kukushkin2005}, as well as that of generalized aggregative games \citep{Jensen2010}. The latter structure allows for the use of the theory of potential games \citep{MondererShapley1996}, specifically, best-response potential games \citep{Voorneveld2000}. The main intuition behind the existence of a differentiable best-response potential function for reciprocal interaction networks is that the potential's Hessian, a symmetric matrix, must correspond to the Jacobian of the system of best-responses, which, in this case, is equal to the network's weight matrix. In this paper, however, we show that the best-response potential structure can be exploited in important classes of networks even when the relations, as expressed by the interaction matrix, are not reciprocal.

We consider convergence to Nash equilibria under one-sided improvement dynamics taking place in discrete time. Starting from a profile of production decisions, in every time period one player receives an opportunity to change her production, while every other player remains on the previous period's level. In the next period, another player receives a revision opportunity, and so on. We consider three versions of the best-response dynamic. In increasing order of generality, these are the pure best-response dynamic, in which every revision takes the player to her current best choice given the actions of others, the best-response-approaching dynamic, in which every revision moves the player into the interval between her current action and her current best choice, and the best-response-centered dynamic, in which every revision reduces the distance between her action and her current best choice. Pure best-response dynamics as described above are widely studied. In directed network games, best-response-approaching dynamics include the naive learning dynamics introduced by \cite{Bervoets2018}, while best-response-centered dynamics are studied in \cite{BHPT2019}. These two dynamics are similar to the directional learning model \citep{SeltenStoecker1986,SeltenBuchta1998} in which players are making attempts to find their targets by adjusting towards the direction they believe the target is located. Such qualitative learning models are known to explain experimental behavior in various settings \citep{CachonCamerer1996,CasonFriedman1997,KagelLevin1999,NagelVriend1999}.

While none of these dynamics can cycle under reciprocal interactions, in directed network games cycles can emerge. Cycles indicate that convergence to the Nash equilibrium is not a universal property of learning processes. We discuss two examples: (1) directed cycle networks allow for best-response cycles as economic activity of the players flows in the opposite direction as the external effects of the network, and (2) parasite-host networks with amplifying links lead to best-response cycles as the parasite's economic activity increases with the host's activity level, while the host's economic activity decreases with that of the parasite.

Nevertheless, classes of networks exist without cycles where convergence to a Nash equilibrium can be shown, given some mild assumptions on the order of updates. In this paper we identify two such classes, networks with transitive relative importance, and games rescalable to exhibit weak influences or weak externalities.

The former class captures networks that can be transformed into symmetric networks with some appropriate rescaling of the action space, an idea raised by \cite{Bramoulle2014}. Rescaling can be understood as changing the measurement of one player's production from, e.g., euros to dollars. Rescaling does not affect the equilibrium structure or the convergence properties of the game, but it does change its nominal interaction structure as expressed by the network's weight matrix. Thus, a network with reciprocal interactions can be rescaled into non-reciprocal ones, which thus inherit its convergence properties. A network can be rescaled in such a way if and only if it satisfies the property of transitive relative importance and it does not have one-way or parasitic interactions.

The relative importance of a link for a player is measured by the relative payoff-effects between the two players. If a link is reciprocal, then the relative importance of that link for both players is unity. A link with a larger weight for one player and a small weight for the other is relatively more important to the former and, inversely, not so important for the latter. Transitive relative importance restricts the network through these values. For instance, if the link $\{i,j\}$ is more important to $i$ than to $j$, and if the link $\{j,k\}$ has equal importance to both participants, then the link $\{i,k\}$ must be more important to $i$ than to $k$. This property presumes a common hierarchy of players with important players whose production matters greatly in relative terms for all individuals and less important players whose production matters little.

This property is closely connected to Kolmogorov's reversibility criterion for Markov chains (see e.g.~\cite{Kelly2011}), as well as transitive matrices \citep{Farkas1999} a property applied in pairwise comparison matrices \citep{Bozoki2010} and Analytic Hierarchy Processes \citep{Saaty1988}. We show that these networks and only these can be rescaled into symmetric ones, and these are the only ones that have a quadratic best-response potential function.

The second class of networks with convergent best-response dynamics are those that are rescalable to exhibit weak influences or weak externalities. A player is influenced weakly by her opponents if the total effects of a unit change in all of her opponents' actions on her are smaller than the effect of a unit change in her own action. These networks are characterized by row diagonally dominant weight matrices. In social networks, this property can be interpreted as a form of individualism. On the other hand, weak external effects are characterized by column diagonally dominant weight matrices, meaning a unit change in any player's action has a larger effect on herself than on all the other players combined. In economics, small level of externalities is a characteristic of efficient markets.

In non-cooperative game theory, weakness of externalities is known be sufficient for uniqueness of the Nash equilibrium as well as convergence to the unique Nash equilibrium under best-response dynamics \citep{GabayMoulin1980,Moulin1986}. In networks, weak influences are studied by \cite{PariseOzdaglar2019} while both types of diagonal dominance are studied by \cite{Scutari2014}. Since in our model, network weights are allowed to be both negative and positive, weak externalities is a separate condition from small network effects \citep{Bramoulle2014,Belhaj2014}. Under small network effects, given an equilibrium production profile, a tremor in a player's production decision is dampened by the network such that the system returns to the original equilibrium under best-response dynamics. We show that small network effects in absolute value, that is, the spectral radius of the absolute value of the weight matrix being less than one is sufficient and necessary for rescalability into networks with weak influences and those with weak externalities. We thus fully characterize the class of networks for which uniqueness of equilibrium and global convergence of best-response dynamics can be shown by either type of diagonal dominance. Moreover, we show that this class of games also has a potential structure.

Additionally, this class also includes the set of directed acyclic networks. In this subclass every interaction is one-way and there are no directed cycles. The latter property imposes a hierarchy on the players with respect to the one-way interactions; players on the highest level of the hierarchy are not affected by any other player, players on intermediate levels are affected by those on higher levels but not by those on lower levels, while players on the lowest level have no effect on any other player. The most direct application of this class of games is pollution management of cities along a river \citep{NiWang2007} or a river network \citep{DongNiWang2012}, but the same model is used in studying the conservation of common-pool resources \citep{RichefortPoint2010} and games with permission structures \citep{vanderBrink2018}.

Overall, our results have a number of implications with respect to convergence in directed network games. A negative finding is that, in networks with directed cycles and parasite-host interactions with amplifying weights, the interpretation of the Nash equilibrium as an outcome of a series of decentralized improvements by the players is questionable as the convergence of simple learning processes is not assured. We complement this with a set of positive results by identifying and characterizing interesting classes of networks where convergence is assured. Our two sets of positive results uncover insights into the relationship between reciprocity of network interactions, the spectral properties of the network, and the games' potential structure as well as generalize the powerful results achieved for the case of reciprocal interactions.

Our paper is organized as follows: Section \ref{sec: model} presents our two main concepts, directed network games and best-response dynamics. In Section \ref{sec: cycle} we present the two simple networks with best-response cycles. Section \ref{sec: transitive} contains our characterization and convergence results for networks that can be rescaled into symmetric networks. Section \ref{sec: weak} contains the same sets of results for networks rescalable to games with weak influences or weak externalities, and shows that this class includes directed acyclic networks. Section \ref{sec: conclusion} concludes.

\section{The model}\label{sec: model}

Let $I = \{1,\ldots,n\}$ be the set of players. For $i\in I$ and upper bounds $\xh_i>0$ the set $X_i = [0, \xh_i]$ is called the action set of player $i$, $X = \prod_{i\in I} X_i$ is called the set of action profiles. We let $x_i\in X_i$ denote the action taken by player $i$ while $x_{-i}\in X \setminus X_i$ denotes the truncated action profile of all players except player $i$, and $x = (x_i)_{i\in I}$ denotes the action profile of all players.

The formal definition of directed network games used in this paper is as follows.

\begin{definition}\label{def: dng}
The tuple $G = (I, X, (\pi_i)_{i\in I})$ is called a \textit{directed network game} with payoff functions $\pi_i\colon X \to\mathbb{R}$ given by
\begin{align}
\label{payoff}
\pi_i(x)=f_i\left(\sum_{j\in I}w_{ij}x_j\right)-c_ix_i,
\end{align}
where $f_i: \mathbb{R}\to\mathbb{R}$ is twice differentiable, $f_i'>0$, $f_i''<0$, $w_{ij}\in\mathbb{R}$, $w_{ii}=1$, and $c_i\geq 0$ for every $i, j\in I$.
\end{definition}

\noindent The interpretation is the following. Each player produces a specialized good with linear production technology, incurring costs $c_i$ for every unit of the good produced. Players derive benefits from the consumption of their own goods and they are affected by their opponents' production decisions. Player $i$'s enjoyment of player $j$'s good is represented by the weight $w_{ij}\in\mathbb{R}$. We normalize the interaction parameter of each player $i$ with herself, $w_{ii}$, to $1$. The overall benefits of player $i$ are given by the benefit function $f_i$ over the weighted sum of her and her opponents' goods. Crucially, we do not impose reciprocal relations, meaning that $w_{ij}\ne w_{ji}$ may hold, so the weight matrix $(w_{ij})_{i,j\in I} = W$ might not be symmetric.

Since the benefit functions $f_i$ are increasing and concave and the cost parameters $c_i$ are positive, for every $x_{-i}$ there is a unique value of $x_i$ that maximizes $\pi_i(x)$. Let the target values $t_i$ be implicitly defined by $f'_i(t_i)=c_i$, i.e. the value player $i$ would produce if all others produce $0$ and let $t=(t_i)_{i\in I}$ denote the vector of targets. We make the simplifying assumption that every player is able to produce her target amount of the good.

\begin{assumption}\label{ass: interior}
We assume that $f_i$ are given such that for every $i\in I$ we have $t_i\in (0, \xh_i)$.
\end{assumption}

Note that if $w_{ij}>0$ and $w_{ji}>0$, then the goods of players $i$ and $j$ are strategic substitutes. If $w_{ij}<0$ and $w_{ji}<0$ then their goods are strategic complements. If $w_{ij}>0$ and $w_{ji}<0$, then we say that players $i$ and $j$ share a parasitic link. If $w_{ij}=0$, then player $i$ is not directly affected by player $j$'s production decision.

For player $i\in I$ and $x\in X$, $b_i(x)=\argmax_{x_i\in X_i}\pi_i(x)$ denotes player $i$'s best-response function.

\begin{lemma}\label{le: br3}
Let $G = (I, X, (\pi_i)_{i\in I})$ be a weighted network game. Then, for every $i\in I$ and $x \in X$ the best response functions are the following:
\begin{align}
\label{best response}
b_i(x) =
\begin{cases}
0 & \text{if } t_i - \sum_{j\in I\setminus \{i\} } w_{ij}x_j < 0, \\
t_i - \sum_{j\in I\setminus \{i\} } w_{ij}x_j & \text{if } t_i - \sum_{j\in I\setminus \{i\} } w_{ij}x_j \in [0, \xh_i], \\
\xh_i & \text{if } t_i - \sum_{j\in I\setminus \{i\} } w_{ij}x_j > \xh_i.
\end{cases}
\end{align}
\end{lemma}

\begin{proof}
First, we calculate the unconstrained best response for player $i$. The first-order condition is
\begin{align*}
\frac{\partial \pi_i(x)}{\partial x} = f_i'\left(\sum_{j\in I}w_{ij}x_j\right)-c_i = 0.
\end{align*}
\noindent Combining with $f'_i(t_i) = c_i$, we get that the unconstrained best response, $\bt_i(x)$ is
\begin{align}
\label{eq: unconstrained}
\bt_i(x) = t_i - \sum_{j\in I\setminus \{i\} } w_{ij}x_j.
\end{align}
\noindent The second-order condition is
\begin{align*}
\frac{\partial^2 \pi_i(x)}{\partial x^2} = f_i''\left(\sum_{j\in I}w_{ij}x_j\right) < 0,
\end{align*}
\noindent hence $\bt_i(x)$ is indeed maximizing the payoff. This means that for every $ x_i > \bt_i(x)$, a marginal increase of $x_i$ decreases $\pi_i(x)$, while for every $ x_i < \bt_i(x)$, a marginal increase of $x_i$ increases $\pi_i(x)$. Therefore, if $\bt_i(x) \in [0, \xh_i]$, then $b_i(x) = \bt_i(x)$. If $\bt_i(x) < 0$, then $b_i(x) = 0$, as choosing a larger $x_i$ would decrease the payoff. Similarly, if $\bt_i(x) > \xh_i$, then $b_i(x) = \xh_i$.
\end{proof}

\begin{proposition}\label{pro: nash}
Every directed network game has at least one Nash equilibrium.
\end{proposition}

\noindent \cite{Bramoulle2014}'s analogue result using Brouwer's fixed-point theorem for a positive and symmetric weight matrix is directly applicable in the directed network case. Let the set of Nash equilibria be denoted by $X^\ast$.



Lemma \ref{le: br3} and Proposition \ref{pro: nash} have close analogues in existing models with symmetric interaction, but, as we will show shortly, results concerning the cycling and convergence of simple learning processes do not hold in general in the directed network case.

We now formally introduce the learning processes of our paper, the best-response dynamic and two extensions.

\begin{definition}\label{def: frequently}
The sequence of action profiles $(x^k)_{k\in \mathbb{N}}$ is a \textit{one-sided dynamic} if
\begin{itemize}
\item for every $k\in \mathbb{N}$ there exists an $i^k\in I$ such that $x_{-i^k}^k = x_{-i^k}^{k+1}$, and
\item for every $i\in I$ the set $K^i=\{k\in\mathbb{N}\colon i^k=i\}$ is infinite.
\end{itemize}
\end{definition}

\begin{definition}[Best-response dynamics]\label{def: br dynamic}
The one-sided dynamic $(x^k)_{k\in \mathbb{N}}$ is a
\begin{itemize}
\item \textit{best-response dynamic} (BRD), if we have  $x_{i^k}^{k+1} = b_{i^k}(x^k)$.
\item \textit{best-response-approaching dynamic} (BRAD) with approach parameter $0\leq\beta<1$, if $|x_{i^k}^{k+1}-b_{i^k}(x^k)|\leq \beta|x_{i^k}^k-b_{i^k}(x^k)|$ and if $x_{i^k}^{k+1}\neq b_{i^k}(x^k)$, then $\sgn(x_{i^k}^{k+1}-b_{i^k}(x^k))=\sgn(x_{i^k}^{k}-b_{i^k}(x^k))$.
\item \textit{best-response-centered dynamic} (BRCD) with centering parameter $0\leq \alpha<1$, if $|x_{i^k}^{k+1}-b_{i^k}(x^k)| \leq \alpha|x_{i^k}^k-b_{i^k}(x^k)|$,
\end{itemize}
for every $k\in\mathbb{N}$.
\end{definition}

\noindent In a one-sided dynamic exactly one player changes her action in every time period. In a BRD, every revision takes the updating player to her best response, in a BRAD, players move closer to their best responses without overshooting it, while in a BRCD, players move closer to their best responses and are allowed to overshoot. The approach and centering parameters of a BRAD and a BRCD, respectively, indicate the maximum fraction to which the distances are allowed to decrease. These processes allow payoff-maximizing players to make mistakes, their ability of reaching the best-response is captured by the two parameters with lower values indicating a higher level of accuracy. It is clear that every BRD is a BRAD and every BRAD is a BRCD.

In all cases, through Definition \ref{def: frequently}, we restrict our attention to dynamics where every player revises infinitely many times. This brings the property that every convergent dynamic will converge to a Nash equilibrium, and no player can get stuck playing a suboptimal action indefinitely by not having the opportunity to revise.

\cite{Bramoulle2014} and \cite{BervoetsFaure2019} consider the BRD in continuous time. \cite{PariseOzdaglar2019}, in addition to the continuous-time dynamic, also considers discrete-time BRD, both simultaneous and sequential updating in a fixed order of players. Our BRD process is more general than the latter as revision opportunities may arrive in any order. \cite{BHPT2019} considers both the BRD and the BRCD as above.

Finally in this section we define one of the main concepts used in this paper, best-response potential games.

\begin{definition}[\cite{Voorneveld2000}]\label{def: br potential}
A game $G = (I, X, (\pi_i)_{i\in I})$ is a \textit{best-response potential game}, if there exists a \textit{best-response potential function} $\phi: X \to\mathbb{R}$ such that for every $i \in I$, and every $x_{-i} \in X_{-i}$ it holds that
\begin{align*}
\argmax_{x_i\in X_i} \pi_i(x) = \argmax_{x_i\in X_i} \phi(x).
\end{align*}

\end{definition}

\noindent Definition \ref{def: br potential} states that $G$ is a best-response potential game if the best-response behavior of all players can be characterized by a single real-valued function $\phi$, called the best-response potential.

\section{Networks with best-response cycles}\label{sec: cycle}
In this section we discuss the cycling of best-response dynamics. The non-existence of best-response cycles is a necessary but not sufficient condition of the convergence of best-response dynamics \citep{Kukushkin2015}, which is true in networks with reciprocal interactions. In this section we show that directed cycle networks and parasitic relations with amplifying links lead to cycling and thus hinder any general convergence results.

We begin with a formal definition of cycles.

\begin{definition}[Cycles]\label{def: br cycle}
A sequence $(x^k)_{k\in \mathbb{N}}$ has a \textit{cycle} if there exist three time periods, $k<k'<k''$ such that $x^k = x^{k''}$, but $x^k\neq x^{k'}$.
\end{definition}

In words, a process has a cycle if it non-trivially revisits an action profile in two different time periods.

\begin{example}[Directed cycle network]\label{ex: circle2}

Consider the directed cycle network, with $I = \{1,2,3\}$, $X_i=[0,1]$ for $i\in I$, and the weight matrix
\begin{align*}
W = \begin{pmatrix}
1 & 0 & 1 \\ 1 & 1 & 0 \\ 0 & 1 & 1
\end{pmatrix},
\end{align*}
\noindent representing a directed cycle with three players. Let $f_i(x) = \log(1+x)$ and $c_i(t_i) = 1/(1+t_i)$ for $i\in I$, so player $i$'s payoff is 
\begin{align*}
\pi_i(x) = \log\left(1+\sum_{j\in I}w_{ij}x_j\right)-\frac{1}{1+t_i}x_i.
\end{align*}
As $f'(t_i) = c_i(t_i)$, $t$ is indeed the vector of targets. Now fix $t = (1, 1, 1)^\top$.

This game has a single Nash equilibrium, $x^\ast=(0.5,0.5,0.5)^\top$. Consider the BRD with the initial action profile $x = (1,0,0)^\top$. Let player 3 receive the first revision opportunity, followed by player 1.

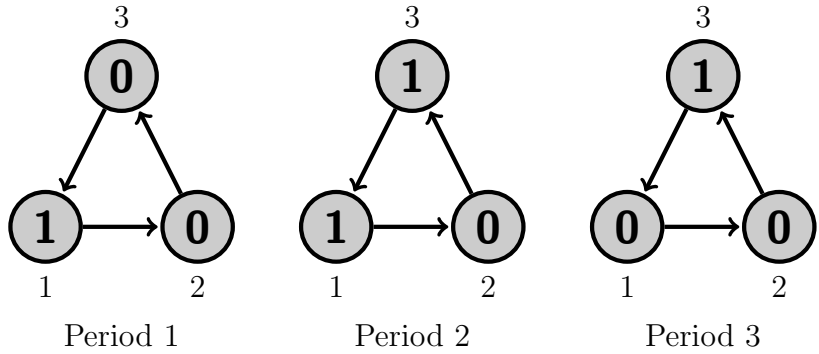
\begin{figure}[!ht]
\centering
\begin{tabular}{ccccc}
\begin{tikzpicture}[shorten >=1pt, auto, node distance=3cm, ultra thick]
    \tikzstyle{node_style} = [circle,draw=black,fill=black!20!,font=\sffamily\Large\bfseries]
    \tikzstyle{edge_style} = [draw=black, line width=2, ultra thick]
\node[node_style, label=below:{$1$}] (v1) at (-1,0) {1};
\node[node_style, label=below:{$2$}] (v2) at (1,0) {0};
\node[node_style, label=above:{$3$}] (v3) at (0,2) {0};
\draw[->]  (v1) -- (v2);
\draw[->]  (v2) -- (v3);
\draw[->]  (v3) -- (v1);
\end{tikzpicture}
&&
\begin{tikzpicture}[shorten >=1pt, auto, node distance=3cm, ultra thick]
    \tikzstyle{node_style} = [circle,draw=black,fill=black!20!,font=\sffamily\Large\bfseries]
    \tikzstyle{edge_style} = [draw=black, line width=2, ultra thick]
\node[node_style, label=below:{$1$}] (v1) at (-1,0) {1};
\node[node_style, label=below:{$2$}] (v2) at (1,0) {0};
\node[node_style, label=above:{$3$}] (v3) at (0,2) {1};
\draw[->]  (v1) -- (v2);
\draw[->]  (v2) -- (v3);
\draw[->]  (v3) -- (v1);
\end{tikzpicture}
&&
\begin{tikzpicture}[shorten >=1pt, auto, node distance=3cm, ultra thick]
    \tikzstyle{node_style} = [circle,draw=black,fill=black!20!,font=\sffamily\Large\bfseries]
    \tikzstyle{edge_style} = [draw=black, line width=2, ultra thick]
\node[node_style, label=below:{$1$}] (v1) at (-1,0) {0};
\node[node_style, label=below:{$2$}] (v2) at (1,0) {0};
\node[node_style, label=above:{$3$}] (v3) at (0,2) {1};
\draw[->]  (v1) -- (v2);
\draw[->]  (v2) -- (v3);
\draw[->]  (v3) -- (v1);
\end{tikzpicture}
\\
Period 1&&Period 2&&Period 3
\end{tabular}
\caption{Economic activity shifts in the opposite direction of externalities under the BRD.}
\label{table: circle2}
\end{figure}
Figure \ref{table: circle2} shows how player 1's initial production moves along its incoming link to player 3. It is easy to see that by allowing player 2 to revise next followed by player 3, activity shifts further on to player 2. Moving along one step further completes the best-response cycle to $x=(1,0,0)^\top$.

The directed cycle of this example thus permits a best-response cycle in which the players' activity moves along the cycle in the opposite direction as the players' externalities. It is easy to see that this property remains true for directed cycles of $n$ players. If $n$ is even, it may be possible to reach an equilibrium in which every second player produces $1$ and the others produces $0$ for some order of revisions, but no $0$-$1$ equilibrium exists if $n$ is odd. Therefore, in the latter case, if every player begins at a production level of either $0$ or $1$, then the only two possible production levels in any best-response path are also $0$ and $1$, and hence, convergence to the Nash equilibrium is impossible. For $n=3$, the game will repeat the cycle in Figure \ref{table: circle2} ad infinitum.
\end{example}

Next, we consider a parasitic link between two players.

\begin{example}[Parasitism]\label{ex: parasite}
Let $I=\{1,2\}$, $f_i(x) = \log(1+x)$, $c_i(t_i)=1/(1+t_i)$ for $i\in I$, $t_1,t_2=1$, and let
\begin{align*}
W = \begin{pmatrix}
1 & -2\\ 0.5 & 1
\end{pmatrix}.
\end{align*}
Under the BRD, starting from the action profile $x=(1,0)^\top$, if both players revise in turns, the game has a best-response cycle of length four.

\begin{figure}[h!]
\centering
\begin{tabular}{cc}
\begin{tikzpicture}[shorten >=1pt, auto, node distance=3cm, ultra thick]
    \tikzstyle{node_style} = [circle,draw=black,fill=black!20!,font=\sffamily\Large\bfseries,minimum size=1cm]
    \tikzstyle{edge_style} = [draw=black, line width=2, ultra thick]
\node[node_style, label=below:{$1$}] (v1) at (-2,0) {1};
\node[node_style, label=below:{$2$}] (v2) at (2,0) {0};
\draw[->]  (-1.4,0.1) -- (1.4,0.1);
\draw[->]  (1.4,-0.1) -- (-1.4,-0.1);
\node[label=below:{$0.5$}] (v3) at (0,0.8) {};
\node[label=above:{$-2$}] (v4) at (0,-0.8) {};
\end{tikzpicture}
&
\begin{tikzpicture}[shorten >=1pt, auto, node distance=3cm, ultra thick]
    \tikzstyle{node_style} = [circle,draw=black,fill=black!20!,font=\sffamily\Large\bfseries,minimum size=1cm]
    \tikzstyle{edge_style} = [draw=black, line width=2, ultra thick]
\node[node_style, label=below:{$1$}] (v1) at (-2,0) {1};
\node[node_style, label=below:{$2$}] (v2) at (2,0) {.5};
\draw[->]  (-1.4,0.1) -- (1.4,0.1);
\draw[->]  (1.4,-0.1) -- (-1.4,-0.1);
\node[label=below:{$0.5$}] (v3) at (0,0.8) {};
\node[label=above:{$-2$}] (v4) at (0,-0.8) {};
\end{tikzpicture}
\\
Period 1: Self-sustaining host. & Period 2: Activation of parasite.
\\
\begin{tikzpicture}[shorten >=1pt, auto, node distance=3cm, ultra thick]
    \tikzstyle{node_style} = [circle,draw=black,fill=black!20!,font=\sffamily\Large\bfseries,minimum size=1cm]
    \tikzstyle{edge_style} = [draw=black, line width=2, ultra thick]
\node[node_style, label=below:{$1$}] (v1) at (-2,0) {2};
\node[node_style, label=below:{$2$}] (v2) at (2,0) {.5};
\draw[->]  (-1.4,0.1) -- (1.4,0.1);
\draw[->]  (1.4,-0.1) -- (-1.4,-0.1);
\node[label=below:{$0.5$}] (v3) at (0,0.8) {};
\node[label=above:{$-2$}] (v4) at (0,-0.8) {};
\end{tikzpicture}
&
\begin{tikzpicture}[shorten >=1pt, auto, node distance=3cm, ultra thick]
    \tikzstyle{node_style} = [circle,draw=black,fill=black!20!,font=\sffamily\Large\bfseries,minimum size=1cm]
    \tikzstyle{edge_style} = [draw=black, line width=2, ultra thick]
\node[node_style, label=below:{$1$}] (v1) at (-2,0) {2};
\node[node_style, label=below:{$2$}] (v2) at (2,0) {0};
\draw[->]  (-1.4,0.1) -- (1.4,0.1);
\draw[->]  (1.4,-0.1) -- (-1.4,-0.1);
\node[label=below:{$0.5$}] (v3) at (0,0.8) {};
\node[label=above:{$-2$}] (v4) at (0,-0.8) {};
\end{tikzpicture}
\\
Period 3: Increased host activity. & Period 4: Passive, free-riding parasite.
\end{tabular}
\caption{The best-response cycle of Example \ref{ex: parasite}.}
\label{fig: parasite}
\end{figure}
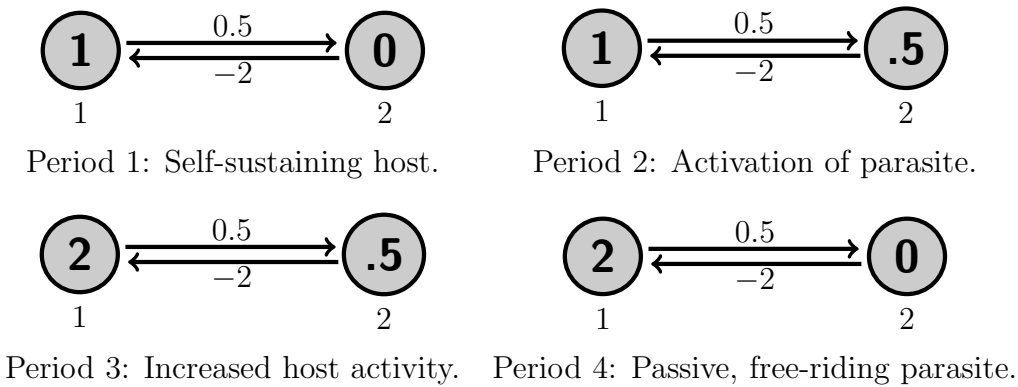

In this example, player $1$ is called the host and player $2$ is called the parasite. A self-sustaining host is engaged by a parasite. The host responds by increasing her activity to offset the parasite's negative effects. The parasite's benefits from the host are large enough that it ceases production entirely and free-rides on the host. The host is then able to return to the self-sustaining stage, completing the cycle. This process is shown in Figure \ref{fig: parasite}.

The same pattern of parasite-host interaction can be replicated by different calibration of parameters. A parasitic network requires $w_{12}<0$ and $w_{21}>0$. If the network weights and targets also satisfy
\begin{enumerate}[label=(\roman*)]
\item $\overline{x}_1\geq t_1-w_{12}(t_2-w_{21}t_1)$ (host has a large enough capacity to feed the parasite),
\item $\overline{x}_2\geq t_2-w_{21} t_1$ (parasite has a large enough capacity to cause a nuisance to a self-sustaining host),
\item $t_2-w_{21}t_1>0$ (self-sustaining host does not satisfy the parasite),
\item $|w_{12}w_{21}|\geq 1$ (the link is amplifying),
\end{enumerate}
then the best-response cycle qualitatively equivalent to that of Figure \ref{fig: parasite} is achieved by sequential updates by the players.
\end{example}

Examples \ref{ex: circle2} and \ref{ex: parasite} together indicate that directed cycles and parasitic relations lead to the cycling of BRDs, and hence that of BRADs and BRCDs.

\section{Transitive relative importance}\label{sec: transitive}
As shown in the previous section, allowing for non-reciprocal interactions in network games changes their convergence properties under BRDs. However, there are classes of directed network games where convergence to the game's Nash equilibrium can still be shown.

It is well known in the literature that convergence of BRDs in networks of reciprocal interactions can be shown by exploiting the game's potential structure. In section VI, \cite{Bramoulle2014} raises the idea that, by an appropriate rescaling of the action space, this method may be extended for some class of directed networks as well. In this section we identify and characterize this class of networks as those with transitive weight matrices. Additionally, we show that, of all directed network games, games played on these and only these networks have quadratic best-response potentials.

We first formally restate the convergence result in the symmetric case for future use. We begin by defining regular updates by players.

\begin{definition}
\label{def: regularly}
We say that players \textit{update regularly} in a one-sided dynamic $(x^k)_{k\in \mathbb{N}}$, if there exists a $K>0$ such that for every $i\in I$ and $k\in\mathbb{N}$ there exists a $k'$ such that $k< k' \leq k+K$ and $i^{k'}=i$.
\end{definition}

\noindent Players updating regularly ensures that no player's frequency of receiving revision opportunities approaches zero.

\begin{proposition}\label{pro: symmetric}
Let $W$ be given such that for every $i,j\in I$ we have $w_{ij}=w_{ji}$ and let $t$ be such that $|X^\ast|<\infty$. Then, every BRD and BRCD $(x^k)_{k\in \mathbb{N}}$ in which players update regularly converges to a Nash equilibrium.
\end{proposition}

\noindent For the full proof of Proposition \ref{pro: symmetric} see Theorem~5.3 in \cite{BHPT2019}. By Lemma~3.2 in \cite{BHPT2019}, for every network, the set of target vectors which result in infinitely many Nash equilibria is small both in the measure theoretic sense and the topologial sense, so the finiteness condition is not restrictive. The main intuition for the BRD part is that $w_{ij}=w_{ji}$ implies that the function

\[x^\top t-\frac{1}{2}x^\top W x\]
\noindent serves as a best-response potential. Every player's update will weakly increase the value of the best-response potential, which is bounded as the action space is bounded. Regular updating ensures that, in time, every player will be close to her current best-response, while finite number of Nash equilibria ensures that convergence of $\phi$ translates to convergence to an isolated peak of the potential value landscape, corresponding to a Nash equilibrium. Additionally, \cite{BHPT2019} show that the finiteness of the equilibrium set is a generic property of undirected network games.

Take a vector $a\in \mathbb{R}^n$, called a scaling vector, such that $a>0$, i.e., for every $i\in I$ we have $a_i> 0$ and let $y_i=a_ix_i$ and $\overline{y}_i=a_i\overline{x}_i$. Furthermore, let $Y_i=[0,\overline{y}_i]$ and $Y=\prod_{i\in I}Y_i$. It is clear that any BRD, BRAD, or BRCD in the game played on $X$ is a BRD, BRAD, or BRCD, respectively, in the game played on $Y$ with the same approach or centering parameters in the cases of BRAD and BRCD. Similarly, the convergence properties of the processes in the game played on $Y$ are identical to those in the game played on $X$.

By Lemma \ref{le: br3} the unconstrained best-response function of player $i$ in the rescaled game is given by
\begin{align*}
t_ia_i-\sum_{j\in I\setminus \{i\} }w_{ij}\frac{a_i}{a_j}y_j.
\end{align*}

Let $w_{ij}a_{i}/a_{j}=v_{ij}$ and let $V=(v_{ij})_{i,j\in I}$ denote the matrix of rescaled weights. Our goal in this section is to characterize the class of networks that are rescalable into a symmetric network, i.e., the set of networks $W$ for which there exists a vector $a>0$ such that for every $i,j\in I$ we have
\begin{align}\label{eq: rescaleable}
\frac{a_i}{a_j}w_{ij}=v_{ij}=v_{ji}=\frac{a_j}{a_i}w_{ji}.
\end{align}
Note that rescaling a network means that we conjugate its weight matrix by the diagonal matrix of scaling weights. That is, consider $S = \mathrm{diag}(a_1, \dots, a_n)$. Then, we have $V=SWS^{-1}$. We continue to talk about rescalability instead of conjugate matrices as $S$ being a diagonal matrix is crucial for our results, as it will be apparent in Theorem \ref{th: transitive}. A similar notion appears in \cite{GolubJackson2012} who rely on this special case of matrix similarity to symmetric matrices.

We now define transitive relative importance of links.

\begin{definition}\label{def: transitive}
We say that the weight matrix $W$ shows \textit{transitive relative importance} if for every $3\leq m\leq n$ and for all pairwise distinct $i_1,i_2,\ldots,i_m\in I$ we have
\[w_{i_1i_2}w_{i_2i_3}\ldots w_{i_{m-1}i_m}w_{i_mi_1}=w_{i_1i_m}w_{i_mi_{m-1}}\ldots w_{i_3i_2}w_{i_2i_1}.\]
\end{definition}

\noindent For simplicity we write that $W$ is transitive if it satisfies Definition \ref{def: transitive}. Notice that if $W$ is symmetric, then it is also transitive. In fact, a symmetric $W$ satisfies the requirements of Definition \ref{def: transitive} for $2\leq m\leq n$. Also notice that if for every $i,j\in I$ it holds that $w_{ij}\neq 0$, then $w_{ij}w_{jk}w_{ki}=w_{ik}w_{kj}w_{ji}$ for every pairwise distinct $i,j,k\in I$ implies transitivity.

The interpretation is as follows: for $w_{ij},w_{ji}\neq 0$ define the value $r_{ij}=w_{ij}/w_{ji}$ as player $i$'s relative importance on the link between $i$ and $j$. It is clear that, if well-defined, the matrix $R=(r_{ij})_{i,j\in I}$ is symmetrically reciprocal. In this case the transitivity property reduces to having
\begin{align}\label{eq: transitive}
r_{ij}r_{jk}=r_{ik},
\end{align}
for every $i,j,k\in I$. Thus, qualitatively, if the link $\{i,j\}$ is more important to $i$ than to $j$ and if the link $\{j,k\}$ is more important to $j$ than to $k$, then the link $\{i,k\}$ must be more important to $i$ than to $k$.

\begin{remark}\label{re: transitive}
Transitivity as defined in Definition \ref{def: transitive} is identical as Kolmogorov's characterization of reversible Markov chains with the exception that network weights are allowed to be negative. If $R$ is defined, then transitivity of $W$ means that $R$ is a consistent pairwise comparison matrix of an Analytic Hierarchical Process \citep{Saaty1988}. Here, the literature seems to use the terms consistent \citep{Bozoki2010} and transitive \citep{Farkas1999} interchangeably when describing the symmetrically reciprocal comparison matrix $R$. We use transitivity of $W$ to indicate the intuition imposed on links.
\end{remark}

The final definition we require is that of sign-symmetry.

\begin{definition}\label{def: sign}
We say that the weight matrix $W$ is sign-symmetric if for every $\{i,j\}\subseteq I$ we have $\sgn(w_{ij}) = \sgn(w_{ji})$.
\end{definition}
\noindent Sign-symmetry of networks rules out one-way interactions and parasitic interactions between players.

We are ready to present the main result of this section.

\begin{theorem}\label{th: transitive}
The following statements are equivalent:
\begin{enumerate}
\item the network $W$ is rescalable into a symmetric matrix,
\item the network $W$ is transitive and sign-symmetric,
\item there exists a game $G$ on $W$ that has a quadratic best-response potential function,
\item every game $G$ on $W$ has a quadratic best-response potential function.
\end{enumerate}
\end{theorem}

\noindent The proof is shown as an appendix.

Theorem \ref{th: transitive} shows that transitivity of a network combined with sign-symmetry is equivalent to it being rescalable into a symmetric network. Additionally, no other network has a quadratic potential function, which shows that this property cannot be exploited further. Our result thus gives a full characterization for which types of directed networks satisfy the requirements put forward in \cite{Bramoulle2014} section VI.

Theorem \ref{th: transitive} and Proposition \ref{pro: symmetric} give rise to the following corollary.

\begin{corollary}\label{cor: transitive}
Let $W$ be a transitive and sign-symmetric network and let $t$ be given such that $|X^\ast|<\infty$. Then, every BRD and BRCD in which players update regularly converges to a Nash equilibrium.
\end{corollary}
\noindent For a fixed symmetric network, since the number of Nash equilibria is finite for almost every target vector \citep{BHPT2019}, Corollary \ref{cor: transitive} also applies to every network and almost every target vector, thus convergence is generically established for this class of networks.

\section{Weak influences and weak externalities}\label{sec: weak}

In this section we characterize another class of networks with convergent dynamics. A key concept in describing this class is the players' influence and the externalities they produce. A player $i$'s decisions are influenced by her opponents through her incoming weights, measured by their total magnitude: $\sum_{j\in I\setminus\{i\}}|w_{ij}|$. Similarly, a player $i$'s external effects on her opponents is measured by the total magnitude of her outgoing weights: $\sum_{j\in I\setminus\{i\}}|w_{ji}|$.

In this section we show that if total influences or externalities are weak enough in a network, then games played on this network have a unique Nash equilibrium and all BRDs and BRADs converge to it. As a motivating example, consider a parametric version of Example \ref{ex: circle2}.

\begin{example}[Directed cycle with weights]\label{ex: circle3}
Consider the three player directed cycle network given by the weight matrix
\begin{align*}
W = \begin{pmatrix}
1 & 0 & \delta \\ \delta & 1 & 0 \\ 0 & \delta & 1
\end{pmatrix},
\end{align*}
\noindent and with $\delta\in[0,1]$. Let $f_i(x) = \log(1+x)$, $c_i(t_i)=1/(1+t_i)$ and fix $t_i=1$ for $i\in \{1,2,3\}$ as previously. By Lemma \ref{le: br3}, the best-response functions are $b_1(x)=1-\delta x_3$, $b_2(x)=1-\delta x_1$, $b_3(x)=1-\delta x_2$. The only Nash equilibrium is $x^\ast=(1/(1+\delta),1/(1+\delta),1/(1+\delta))^\top$.

If the strength of the influences is as strong as the own effects, i.e. $\delta=1$, then best-response cycles may exist as shown in Example \ref{ex: circle2}. In a game with no interaction, $\delta=0$, best-response cycles cannot exist, since the matrix is symmetric. We now consider how the cycling properties change by changing $\delta$.

In Table \ref{tab: circle3} we show the sequence of action profiles in the BRD where players receive revision opportunities in the same, repeating order $(3,1,2)$, starting, again, in the action profile $(1,0,0)^\top$. In this order of revisions, the player holding the revision opportunity in period $k$ will revise to $\sum_{\ell=0}^k(-\delta)^\ell=(1-(-\delta)^{k+1})/(1+\delta)$ for parameter values of $\delta<1$ (see Table \ref{tab: circle3}). Playing on in this order will produce no cycles, and lead to convergence to the Nash equilibrium.

\begin{table}
\centering
\begin{tabular}{cccccccc}
\toprule
$k$ & $x_1^k$ & $x_2^k$ & $x_3^k$ & $i^k$ & $\sum_{j\in I}w_{i^kj}x_j$ &$b_{i^k}(x_{i^k})$ \\
\hline
$0$ & $1$ & $0$ & $0$ & $3$&$0$&$1$\\
$1$ & $1$ & $0$ & $1$ & $1$&$\delta$&$1-\delta$\\
$2$ & $1-\delta$ & $0$ & $1$ & $2$&$\delta-\delta^2$&$1-\delta+\delta^2$\\
$3$ & $1-\delta$ & $1-\delta+\delta^2$ & $1$ & $3$&$\delta-\delta^2+\delta^3$&$1-\delta+\delta^2-\delta^3$\\
\bottomrule
\end{tabular}
\caption{The best-response dynamic of Example \ref{ex: circle3}.}
\label{tab: circle3}
\end{table}
\end{example}

Example \ref{ex: circle3} suggests that even a non-negligible level of influences/externalities can lead to the disappearance of best-response cycles. As we will show in this section, this turns out to be a general property: if every player's total influences or total externalities are smaller than her own weight, the game has a single Nash equilibrium and every BRD and BRAD converges to it.

We introduce these games formally.

\begin{definition}\label{def: weak}
A network $W$ has 
\begin{itemize}
\item \textit{weak influences} if for every $i\in I$ it holds that $\sum_{j\in I\setminus\{i\}}|w_{ij}|<1$,
\item \textit{weak externalities} if for every $i\in I$ it holds that $\sum_{j\in I\setminus\{i\}}|w_{ji}|<1$.
\end{itemize}
\end{definition}

Networks with weak influences are characterized by \textit{row diagonally dominant} weight matrices, while those with weak externalities have \textit{column diagonally dominant} weight matrices. Games with weak influences satisfy Assumption 2b of \cite{PariseOzdaglar2019}, while both classes are covered by \cite{Scutari2014}'s Proposition 7.

It turns out that both classes of games have a unique Nash equilibrium and every BRD and BRAD converges to it. Once again we can make use of rescaling: any network which is rescalable into one of the two classes inherits the uniqueness of the Nash equilibrium as well as the convergence properties. As before, for $a\in \mathbb{R}^n$, $a>0$, define $V=(v_{ij})_{i,j\in I}$ as $v_{ij}=w_{ij}a_i/a_j$. We show that the two classes are rescalable into each other. Furthermore, a network is rescalable to either class if and only if the spectral radius of $|W|-I_n$ is less than one, where $I_n$ is the $n\times n$ identity matrix.

Recall that the \emph{spectral radius} $\rho(M)$ of a square matrix $M\in \C^{n\times n}$ is the largest absolute value of its eigenvalues, i.e.,
\[\rho(M)=\max\{ |\lambda | : \lambda\in \C \text{ is an eigenvalue of } M\}.\]

The next proposition states that a network $W$ can be rescaled into one with weak influences and one with weak externalities if and only if $\rho(|W|-I_n)<1$.

\begin{proposition}
\label{pro: weakspectral}
The following statements are equivalent for $W$.
\begin{enumerate}
\item There exists a scaling vector $a\in \R^n$, $a>0$ such that the rescaled matrix $V$ has weak influences.
\item There exists a scaling vector $a\in \R^n$, $a>0$ such that the rescaled network $V$ has weak externalities.
\item $\lim_{k\to\infty} (|W|-I_n)^k=0$.
\item $\rho(|W|-I_n)<1$. 
\end{enumerate}
\end{proposition}

The fact that a diagonally dominant $W$ satisfies points 3 and 4 of Proposition \ref{pro: weakspectral}, as well as the equivalence of 3 and 4 are well-known in linear algebra. Our characterization adds the notion of rescalability.

\begin{proposition}
\label{pro: weaknash}
If $\rho(|W|-I_n)<1$, then every game played on the network $W$ has a unique Nash equilibrium.
\end{proposition}

By \cite{Ui2016}, every positive definite $W$ has a unique equilibrium. If $W$ is symmetric, then the condition $\rho(|W|-I_n) < 1$ implies positive definiteness, but for an asymmetric $W$ neither condition is implied by the other.

Proposition \ref{pro: weaknash} follows from the fact that such networks are rescalable to diagonally dominant matrices (Proposition \ref{pro: weakspectral}), allowing us to use \cite{Moulin1986}, Chapter 6, Theorems 2 and 3. As a result, games played on such networks are \emph{dominance solvable} \cite{Moulin1984}, that is, iterated elimination of dominated strategies leads to a unique solution. Thus, they have a unique Nash equilibrium and convergent BRD. We now show a stronger convergence property: the convergence of BRAD. 

\begin{theorem}
\label{thm: weak}
If $\rho(|W|-I_n)<1$, then in every game played on network $W$, every BRAD (and hence every BRD) converges to the unique Nash equilibrium.
\end{theorem}

Theorem \ref{thm: weak} is related to Theorem 4 of \cite{Moulin1986}, Chapter 6, which shows local stability of equilibria with respect to BRD for which the Jacobian's spectral radius is less than one. Our condition is global due to the linear best response functions of the network environment hence we have global convergence. Another related result is Theorem 4.1 of \cite{GabayMoulin1980} who shows convergence of BRAD where revisions arrive in a fixed order for diagonally dominant Jacobians. Our result is more general in directed network games as we cover a wider class of networks and in that revision opportunities may arrive in any order.

Finally, we show that games played on such networks are best-response potential games.
 
\begin{proposition}
\label{pro: weakpot}
Consider a game played on a network $W$, and assume that $a\in \R^n$, $a>0$ is a scaling vector such that the rescaled network $V$ has weak externalities. Then, the function
\[\phi'(x)=-\sum_{i\in I}a_i|x_i-b_i(x)|\]
\noindent is a best-response potential function of the game.
\end{proposition}

\noindent The proofs of Propositions \ref{pro: weakspectral} and \ref{pro: weakpot}, as well as that of Theorem \ref{thm: weak} are shown as an appendix.

As demonstrated by Example \ref{ex: circle2}, a spectral radius equal to $1$ leads to best-response cycles, hence these results are tight. Additionally, notice that the BRCD may lead to cycles in this gameclass. For instance, the BRD shown in Example \ref{ex: circle2} is a BRCD in Example \ref{ex: circle3} for $\delta=0.9$.

It is easy to show that a necessary condition of a network being rescalable to one with weak externalities is that no pair of players have an amplifying link, i.e., one where the product of weights is larger than the size of the own effects in absolute value.

\begin{lemma}\label{le: amplifying}
Let $W$ be given such that there exists an $a\in\mathbb{R}^n$, $a>0$ for which the rescaled network $V$ is with weak externalities. Then, for every $i,j\in I$ it holds that $|w_{ij}w_{ji}|<1$.
\end{lemma}

\begin{proof}
Suppose that $|w_{ij}w_{ji}|>1$. Since $V$ has weak externalities we must have $v_{ij},v_{ji}<1$. However,
\[|v_{ij}v_{ji}|=|\frac{a_i}{a_j}w_{ij}\frac{a_j}{a_i}w_{ji}|>1,\]
\noindent a contradiction.
\end{proof}

\noindent By Definition \ref{def: weak} it is clear that unlike networks that are rescalable to be symmetric, those that are rescalable to exhibit weak influences or externalities may allow one-way interactions as well as parasitic ones. However, in the latter case, the link cannot be amplifying (Lemma \ref{le: amplifying}), which rules out the best-response cycles seen in Example \ref{ex: parasite}. Additionally, as seen in Example \ref{ex: circle2}, if the network contains directed cycles we once again get best-response cycles.

Without directed cycles, however, the network can always be rescaled into one with weak externalities. We now formalize this fact starting with the definition of directed acyclic networks.

\begin{definition}\label{def: dan}
A network $W$ is called a \textit{directed acyclic network (DAN)} if for every $i<j$, $i,j\in I$ we have $w_{ij}=0$.
\end{definition}

\noindent In other words, if $W$ is lower triangular, then the game is played on a \textit{DAN}. One can think of an equivalent characterization with upper triangular matrices, or any permutation of players which leaves $W$ as a lower triangular matrix. In Definition \ref{def: dan} players with lower indices are higher up in the hierarchy, i.e. player $1$ is unaffected by any other player's action, player $2$ is only affected by player $1$, etc.

An economic application of this game is pollution management of a number of cities with industrial zones located along a river. Each city decides on the amount of money spent on cleaning the industrial waste in the river and their decisions affect only those cities that are located downstream. The cities' target values describe the point at which the marginal benefits of an extra dollar's worth of cleaner water are the same as the costs for that city. Models of this problem include \cite{NiWang2007} and \cite{DongNiWang2012}.

Games of the above nature introduce a hierarchy of players. Such hierarchies are present in most production chains where goods -- and therefore externalities -- flow downstream, or in some organizational networks such as the military where orders are traveling down the chain of command. Directed acyclic cycles have applications in biology as well among many other fields; most trophic networks also have hierarchical features with apex predators on the highest level and prey animals on lower levels.

We now show that every game played on a DAN is rescalable into one with weak externalities.

\begin{proposition}\label{pro: dan}
For every DAN, $W$ there exists $a\in\mathbb{R}^n$, $a>0$ such that the rescaled network, $V$ has weak externalities.
\end{proposition}

\begin{proof}
Let $W$ be a DAN, i.e., we have $w_{ij}=0$ for every $i<j$. First, let $a_n=1$. We will define the rest of the $a_i$'s recursively. Assume that we have already defined $a_n$, $a_{n-1}, \dots, a_{n-j+1}$. Let us choose $a_j>0$ so that we have
\begin{align}
\sum_{k=j+1}^n a_k |w_{kj}| < a_j. \label{eq: DAN}
\end{align}
Now we can verify that the rescaled matrix $V$ has weak externalities. Let $j\in I$. We have
\begin{align*}
\sum_{i\in I\setminus \{j\}} |v_{ij}| &=\sum_{i\in I\setminus\{j\}} |w_{ij}|\frac{a_i}{a_j}\\
&= \sum_{i=1}^{j-1} 0\cdot \frac{a_i}{a_j} + \frac{1}{a_j} \sum_{i=j+1}^n |w_{ij}| a_i\\
&< 0 +  \frac{1}{a_j}\cdot a_j=1 \qquad \text{ by (\ref{eq: DAN}).}
\end{align*}
This concludes the proof.
\end{proof}

\noindent The following corollary regarding the convergence of BRD and BRAD is implied by Theorem \ref{thm: weak}.

\begin{corollary}\label{cor: dan}
For every game played on a \textit{DAN}, every BRD and BRAD converges to the game's unique Nash equilibrium.
\end{corollary}

\section{Conclusion}\label{sec: conclusion}

In this paper we analyze directed network games, a generalization of the private provision of public goods games model to include possibly non-reciprocal relationships. These cover one-way interactions, unequal interactions, and parasitism. While weighted networks and simple graphs are very useful frameworks, more nuanced models of social and economic networks should include non-reciprocal interactions.

While best-response dynamics on games played on symmetric networks are known to converge to a Nash equilibrium due to the games' potential structure, this is not true in general for networks with asymmetric weight matrices. In this paper we show that both one-way interactions and parasitic interactions can create best-response cycles. This questions the interpretation of the Nash equilibrium as the result of individual improvements by the players. Together with other known problems of the Nash equilibrium both conceptual and behavioral, equilibrium analysis of such games may be of questionable value in settings with possibly non-reciprocal interactions.

There are classes of asymmetric networks, however, where the predictive power of the Nash equilibrium is retained. In this paper we highlight two such classes; those that can be rescaled into symmetric networks and those that are rescalable to networks with weak influences or weak externalities. We characterize the former type by transitive relative importance of players and sign-symmetry of the weight matrix. Additionally, this class captures all networks with quadratic best-response potential functions, indicating that other network types with convergence require different approaches to identify. 

The latter class captures individualistic social networks as well as situations where the economic externalities have been internalized. We show that these types are equivalent with respect to rescaling and any network with a spectral radius less than one is rescalable to either. Such games are best-response potential games, have a unique Nash equilibrium, and all BRDs and BRADs converge to it. A necessary condition for a network to be rescalable to one with weak externalities is the absence of amplifying links.

Directed acyclic networks can always be rescaled into networks with weak externalities. This subclass imposes a hierarchy on the players; players are only influenced by opponents on higher levels and only influence those on lower levels. Directed acyclic networks have an established application in economics in the pollution management of cities along a river or a river network, but, since the hierarchical structure is widely studied, there are other potential fields of application.

Our results unlock a number of insights into network games. The most apparent general result is a negative one: the convergence properties of games played on symmetric networks do not generalize well for the asymmetric case. For directed cycles and parasitic interactions best-response cycles may appear, thus identifying convergent classes of networks that include any of these types of interactions are likely to require different methodologies than the best-response potentials, rescaling, and spectral properties used in this paper.

Our positive contribution consists of the full exploration of the idea of rescalability into symmetric matrices as well as the identification of weak influences/exter\-nalities as networks with convergent dynamics and the full characterization of the latter class. On a technical level we identify all network games with a quadratic potential structure and unlock a new class with a different potential structure. We thus broaden the set of sufficient conditions that guarantee convergence in network games. Finding a set of sufficient and necessary conditions, or, failing that, broader sets of sufficient conditions, is an important direction left for future research.

\appendix
\section{Appendix}
\subsection*{Proofs for Section \ref{sec: transitive}}

\subsubsection*{Theorem \ref{th: transitive}}
It is clear that statement 4 implies statement 3. We show the remaining three implications.

\begin{proof}[Proof of 1 $\Rightarrow$ 4.]
Suppose that the matrix $W$ can be rescaled into a symmetric matrix $(v_{ij})_{i,j\in I}$ with the vector $a\in \mathbb{R}^n$, $a>0$, i.e., we have $v_{ij}=w_{ij}a_i/a_j$. Let $G$ be any game on $W$. We now show that the following quadratic function is a best-response potential of the game.
\begin{align*}
    \phi^Q (x)= \sum_{i\in I}a_i^2x_it_i - \frac{1}{2}\sum_{i\in I}\sum_{j\in I}x_ix_ja_ia_jv_{ij}.
\end{align*}
For every $i\in I$, the partial derivative is as follows:
\begin{align*}
    \frac{\partial\phi^Q}{\partial x_i} (x) &= a_i^2t_i - \sum_{j\in I}x_ja_ja_iv_{ij}\\
    &= a_i^2t_i - \sum_{j\in I}x_ja_i^2w_{ij}\\
    &= a_i^2\left(t_i - \sum_{j\in I}x_jw_{ij}\right)\\
    &= a_i^2\cdot \tilde b_i(x)
\end{align*}
We used that $v_{ij}= v_{ji}$ and that $a_ja_iv_{ij} = w_{ij}a_i^2$.

Also, we have 
\begin{align*}
\frac{\partial^2 \phi^Q}{\partial x_i^2}(x)=-a_i^2w_{ii}=-a_i^2 <0.
\end{align*}
Therefore, if $b_i(x)=\tilde b_i(x)$, then $\phi^Q(x)$ is maximal when $\tilde b_i(x)=b_i(x)=x_i$. If $\tilde b_i(x)<0$ for some $x_{-i}$, then the derivative of $\phi^Q$ with respect to $x_i$ is uniformly negative on $[0,\xh_i]$, so the maximum is achieved when $x_i=0$. On the other hand, when $\xh_i< \tilde b_i(x)$, then the derivative of $\phi^Q$ with respect to $x_i$ is positive on $[0,\xh_i]$, and hence it takes its maximum in $\xh_i$.

Therefore, for a fixed $x_{-i}$ vector the function $\phi^Q$ is maximal when player $i$ is in her best response, so $\phi^Q$ is a best-response potential function.
\end{proof}

\begin{proof}[Proof of 3 $\Rightarrow$ 2.]
Suppose that there exists a quadratic best-response potential function, $\phi^Q$, for a game $G$ on $W$. We show that the matrix $W$ is sign symmetric and transitive.

Let the potential function be given as follows.
\begin{align*}
    \phi^Q (x)= \sum_{i\in I} p_ix_i - \frac{1}{2}\sum_{i\in I} q_{ii}x_i^2 - \sum_{i,j\in I, i>j} q_{ij}x_ix_j,
\end{align*}
where $p_i\in \mathbb{R}$ for every $i\in I$ and $q_{ij} \in \mathbb{R}$ for every $i,j\in I$, $i\geq j$.
For $j>i$, we set $q_{ij} = q_{ji}$ for convenience. With this notation, the partial derivative of $\phi^Q$ is:
\begin{align}\label{eq: phiq}
    \frac{\partial\phi^Q}{\partial x_i} (x)= p_i - q_{ii}x_i - \sum_{j\in I\setminus \{i\}}q_{ij}x_j .
\end{align}

The function $\phi^Q$ is a best-response potential of the weighted network game $G$, so for every $x\in X$ and for every $i\in I$, the partial derivative of $\phi^Q$ is zero exactly when player $i$ is in her best response. Note that we have $0<t_i<\xh_i$ for every $i\in I$, this means that $t_i=b_i(0)=\tilde b_i(0)\in (0,\xh_i)$ for every $i\in I$. Therefore, there exists a neighborhood of $0$ where each player's unconstrained best response is equal to her best response. Let $\eps >0$ be so that for every $x\in [0,\eps]^n$ and for all $i\in I$ we have $b_i(x)=\tilde b_i(x)$. For a fixed $i$, the functions $\frac{\partial\phi^Q}{\partial x_i}(x)$ and $\tilde b_i(x)-x_i$ are both linear in $x$ and they have the same zero set when $x\in [0,\eps]^n$. Hence, they must be equal up to a constant factor: there exists $d_i\neq 0$ such that we have
\begin{align}\label{eq: pot=br}
    t_i - \sum_{j\in I} w_{ij} x_j &= \tilde b_i(x)-x_i =d_i\frac{\partial \phi^Q}{\partial x_i} (x)= d_i\left( p_i -\sum_{j\in I\setminus \{i\}} q_{ij}x_j -q_{ii}x_i\right)
\end{align}

Additionally, as we are maximizing $\phi^Q$, the second derivative of $\phi^Q$ with respect to $x_i$ has to be negative, so $q_{ii}>0$ for every $i\in I$. From (\ref{eq: pot=br}), we get the following for every $i,j\in I$.
\begin{align}
    t_i &= d_i\cdot p_i \nonumber\\
    w_{ii}=1 &= d_i\cdot q_{ii} \nonumber\\
    w_{ij} &= d_i\cdot q_{ij} \label{eq: third}
\end{align}
Since $q_{ii}>0$, we must have $d_i>0$ as well for all $i\in I$. 

The first two equations give no constraints for $W$. We have to show that if (\ref{eq: third}) holds for all $i,j\in I$, that implies the transitivity and the sign-symmetry of the weight matrix $W$.

Since the $d_i$'s are positive and $q_{ij}=q_{ji}$, we have
\[ \sgn(w_{ij})=\sgn(d_iq_{ij})=\sgn(q_{ij})=\sgn(q_{ji})=\sgn(d_jq_{ji}) =\sgn(w_{ji}),\]
so the matrix $W$ is sign-symmetric.

Now let $3\leq s\leq n$, then for all $i_1,i_2,...,i_s \in I$ pairwise distinct, we have 
\begin{align*}
	w_{i_1i_2}w_{i_2i_3}\dots w_{i_{s-1}i_s}w_{i_si_1} &= (d_{i_1}q_{i_1i_2})(d_{i_2}q_{i_2i_3})\dots (d_{i_{s-1}}q_{i_{s-1}i_s})(d_{i_s}q_{i_si_1})\\
	&= (d_{i_1}d_{i_2}\dots d_{i_s})(q_{i_1i_2}q_{i_2i_3}\dots q_{i_{s-1}i_s}q_{i_si_1})\\
	&= (d_{i_1}d_{i_2}\dots d_{i_s})(q_{i_2i_1}q_{i_3i_2}\dots q_{i_si_{s-1}}q_{i_1i_s})\\
	&= (d_{i_2}q_{i_2i_1})(d_{i_3}q_{i_3i_2})\dots (d_{i_s}q_{i_si_{s-1}})(d_{i_1}q_{i_1i_s})\\
	&= w_{i_2i_1}w_{i_3i_2}\dots w_{i_si_{s-1}}w_{i_1i_s}
\end{align*}
using $w_{ij} = d_iq_{ij}$ for every $i,j\in I$. Therefore, the matrix $W$ is transitive.
\end{proof}

\begin{proof}[Proof of 2 $\Rightarrow$ 1.]

Suppose that the matrix $W$ is transitive and sign-symmetric. We would like to find a scaling vector $a\in \mathbb{R}^n$, $a>0$, such that the rescaled matrix is symmetric, i.e., for every pair $i,j\in I$ we have $w_{ij}a_i/a_j=w_{ji}a_j/a_i$.

First, let us assume that the graph of $W$ is connected, so there exists a path between any two players. If the graph is not connected, we follow the described algorithm for every connected component of the graph separately in order to define the vector $a$.

We start by ordering the players in the following way. Choose player 1 arbitrarily. Let the neighbors of player 1 be $2,\dots, n_1$. Players $n_1+1, \dots, n_2$ are those neighbors of 2 that are not neighbors of 1, and so on: players $n_k+1, n_k+2,\dots, n_{k+1}$ are those neighbors of player $k+1$ that are not neighbors of players $1,2,\dots, k$. (If there are no such players, then we have $n_k=n_{k+1}$.) Due to the sign-symmetry of $W$, if two players are neighbors, there is a directed edge between them in both directions. Hence, since $W$ is connected, there exists a directed path from player 1 to every player, so after at most $n$ steps we reach all players.

Let $a_1=1$. We define all $a_j$'s for $j\geq2$ recursively in the following way. Suppose we have already defined $a_1, a_2,...,a_{j-1}$, and $n_{i-1} < j \leq n_i$, so player $j$ is a neighbor of $i$ but not of players $1,...,i-1$. In other words, $i$ is the neighbor of $j$ with the smallest index. Clearly $i<j$, so $a_i$ is already defined. Also, we have $w_{ji}\neq 0\neq w_{ij}$, since $i$ and $j$ are neighbors. Let
\begin{align}
    a_j = a_i \frac{\sqrtabs{w_{ij}}}{\sqrtabs{w_{ji}}}.
\end{align}

Now we show that with this scaling vector $a$, for any $k,\ell\in I$ we have $v_{k\ell} = v_{\ell k}$. Take any $k,\ell\in I$. If $k=\ell$, then $w_{k\ell}=v_{k\ell}=v_{\ell k}=w_{\ell k}=1$, so we can assume that $k\neq\ell$. Let $i_1\in I$ be so that $k$ is a neighbor of $i_1$, but not of players $1,\dots, i_1-1$, so $n_{i_1-1}< k\leq n_{i_1}$. Similarly, let $j_1\in I$ be so that $n_{j_1-1}< \ell \leq n_{j_1}$. For every $m\in \mathbb{N}$, we define $i_m$ and $j_m$ recursively:
\begin{align*}
i_{m+1} &=\min\{i\in I : w_{i_mi}\neq 0\}, \\
j_{m+1} &= \min\{j\in I : w_{j_mj}\neq 0\}.
\end{align*}
Both sequences are decreasing, and they both stabilize when they reach 1.
Let $r, s\in \mathbb{N}$ be minimal so that $i_r=j_s$. Note that this is going to happen, the latest when they are both equal to 1, and also note that from this point on, the two sequences coincide. Therefore, we have $k>i_1>\dots>i_r$ and $\ell>j_1>\dots >j_s$, and the numbers $k, i_1, \dots, i_{r-1}, \ell, j_1,\dots, j_{s}$ are distinct.

We can calculate $a_k$ as follows.
\begin{align*}
a_k &= a_{i_1}\frac{\sqrtabs{w_{i_1k}}}{\sqrtabs{w_{ki_1}}}\\
 &= a_{i_2}\frac{\sqrtabs{w_{i_2i_1}}}{\sqrtabs{w_{i_1i_2}}}\frac{\sqrtabs{w_{i_1k}}}{\sqrtabs{w_{ki_1}}}\\
 & \ \ \vdots \\
 &= a_{i_r}\frac{\sqrtabs{w_{i_ri_{r-1}}}}{\sqrtabs{w_{i_{r-1}i_r}}}\dots 
 \frac{\sqrtabs{w_{i_2i_1}}}{\sqrtabs{w_{i_1i_2}}}\frac{\sqrtabs{w_{i_1k}}}{\sqrtabs{w_{ki_1}}}\\
 &= a_{i_r}\left| \frac{ w_{i_ri_{r-1}}\dots w_{i_2i_1}w_{i_1k}}{w_{i_{r-1}i_r}\dots w_{i_1i_2}w_{ki_1}} \right|^{\frac{1}{2}}.
\end{align*}

Similarly, we have
\begin{align*}
a_{\ell}= a_{j_s}\left| \frac{ w_{j_sj_{s-1}}\dots w_{j_2j_1}w_{j_1\ell}}{w_{j_{s-1}j_s}\dots w_{j_1j_2}w_{\ell j_1}} \right|^{\frac{1}{2}}.
\end{align*}

Now we can compute $v_{k\ell}$ using that $i_r=j_s$.
\begin{align*}
v_{k\ell} &= \frac{a_k}{a_{\ell}}w_{k\ell}=a_k (a_{\ell})^{-1}w_{k\ell} \\
&= \left( a_{i_r}\left| \frac{ w_{i_ri_{r-1}}\dots w_{i_2i_1}w_{i_1k}}{w_{i_{r-1}i_r}\dots w_{i_1i_2}w_{ki_1}} \right|^{\frac{1}{2}} \right)
\left( a_{j_s}\left| \frac{ w_{j_sj_{s-1}}\dots w_{j_2j_1}w_{j_1\ell}}{w_{j_{s-1}j_s}\dots w_{j_1j_2}w_{\ell j_1}} \right|^{\frac{1}{2}} \right)^{-1} w_{k\ell}\\
&=\left| \frac{(w_{i_ri_{r-1}}\dots w_{i_2i_1}w_{i_1k})(w_{j_{s-1}j_s}\dots w_{j_1j_2}w_{\ell j_1})}{(w_{i_{r-1}i_r}\dots w_{i_1i_2}w_{ki_1})(w_{j_sj_{s-1}}\dots w_{j_2j_1} w_{j_1\ell})} \right|^{\frac{1}{2}} w_{k\ell}\\
&= \left| \frac{w_{i_1k} w_{i_2i_1}\dots w_{i_ri_{r-1}} w_{j_{s-1}j_s} \dots w_{j_1j_2} w_{\ell j_1}}{w_{ki_1}  w_{i_1i_2} \dots w_{i_{r-1}i_r}  w_{j_sj_{s-1}}\dots w_{j_2j_1} w_{j_1\ell} } \right|^{\frac{1}{2}}  \left| \frac{w_{k\ell}}{w_{\ell k}} \right|^{\frac{1}{2}} \cdot \sgn(w_{k\ell})\sqrtabs{w_{k\ell}w_{\ell k}}\\
&= \left| \frac{w_{i_1k} w_{i_2i_1}\dots w_{j_si_{r-1}} w_{j_{s-1}j_s} \dots w_{j_1j_2} w_{\ell j_1} w_{k\ell} }{w_{ki_1}  w_{i_1i_2} \dots w_{i_{r-1}j_s}  w_{j_sj_{s-1}}\dots w_{j_2j_1} w_{j_1\ell} w_{\ell k}} \right|^{\frac{1}{2}} \sgn(w_{k\ell})\sqrtabs{w_{k\ell}w_{\ell k}}
\end{align*} 
By the transitivity assumption for $k, i_1,\dots i_{r-1}, j_s, \dots j_1, \ell \in I$, we have
\[ w_{ki_1}  w_{i_1i_2} \dots w_{i_{r-1}j_s}  w_{j_sj_{s-1}}\dots w_{j_1\ell} w_{\ell k} = w_{i_1k} w_{i_2i_1}\dots w_{j_si_{r-1}} w_{j_{s-1}j_s} \dots w_{\ell j_1} w_{k\ell}, \]
and hence  $v_{k\ell}= \sgn(w_{k\ell}) \sqrtabs{w_{k\ell}w_{\ell k}}$. Similarly, we have $ v_{\ell k} =sgn(w_{\ell k}) \sqrtabs{w_{\ell k}w_{k\ell}}$. Since $W$ is sign-symmetric, these two numbers are equal, so $v_{k\ell}=v_{\ell k}$ for all $k,\ell \in I$.

Finally, we show that the statement holds even if the graph of $W$ is not connected. In this case we order the players in each component and define the $a_i$'s for the components separately as described in the beginning of the proof. Now take any $i,j\in I$. If $i$ and $j$ are the same connected component, we have already shown that $v_{ij} = v_{ji}$. If they are in different components, we know that $w_{ij} = w_{ji} = 0$, therefore $v_{ij} = 0 = v_{ji}$. This concludes the proof.
\end{proof}

\subsection*{Proofs for Section \ref{sec: weak}}

Let us introduce some notations and terminology that will be used in the proofs of this section.

Recall that a matrix norm $\|\cdot \|\colon \C^{m\times n}\to \R_+$ is an \emph{induced norm} if it is induced by vector norms on $\C^m$ and $\C^n$, i.e., there exist norms $\| \cdot \|_{\C^m}\colon \C^m\to \R_+$ and $\| \cdot \|_{\C^n}\colon \C^n\to \R_+$ such that for $M\in \C^{m\times n}$ we have
\[ \| M\|=\sup \{ \|Mx\|_{\C^m} : x\in \C^n, \|x\|_{\C^n}=1\}.\]
This definition implies that we have
\begin{align}
\|Mx\|_{\C^m} \leq \|M\| \cdot \|x\|_{\C^n} \label{eq: inducednorm}
\end{align}
for every $M\in \C^{m\times n}$ and $x\in \C^n$.

We will use a variation of the $\infty$-norm: the weighted maximum norm. Fix a weight vector $u\in \R^{n}$, $u>0$. For $x\in \C^n$, let
\[\|x\|_{\infty}^u=\max\{ |x_i|/u_i : 1\leq i\leq n\}.\]
The induced matrix norm is the following: for $M=(m_{ij})\in \C^{n\times n}$, we have
\begin{align}
\|M\|_{\infty}^u=\sup\{ \|Mx\|_{\infty}^u : \|x\|_{\infty}^u=1 \}=\max\left\{ \frac{1}{u_i}\sum_{j=1}^n u_j|m_{ij}| : 1\leq i\leq n\}\right\}.\label{eq: inftynorm}
\end{align}
For a vector $u\in \R^n$, $u>0$, let $u^{-1}\in \R^n$ denote the vector with entries $u_i^{-1}$. Let $D_u=\mathrm{diag}(u_1,u_2,\dots,u_n)\in \R^{n\times n}$, by (\ref{eq: inftynorm}), we have
\begin{align*}
\|M\|_{\infty}^u=\|D_u^{-1} MD_u\|_{\infty}.
\end{align*}
If $W$ is a network and $a\in \R^n$, $a>0$ is a scaling vector, then for the rescaled matrix $V$ we have $v_{ij}=w_{ij}a_i/a_j$. Equivalently, $V=D_aWD_a^{-1}$. Therefore, we have
\begin{align}
\|W\|_{\infty}^{a^{-1}}=\| D_{a^{-1}}^{-1} W D_{a^{-1}}\|_{\infty}= \| D_a WD_a^{-1}\|_{\infty}=\| V\|_{\infty}. \label{eq: WVnorm}
\end{align}
For a matrix $M=(m_{ij})_{1\leq i,j\leq n}\in \C^{n\times n}$, we will denote by $|M|\in \R_+^{n\times n}$ the matrix $(|m_{ij}|)_{1\leq i,j\leq n}$.

\noindent We will use the following statement.

\begin{proposition}[Perron-Frobenius Theorem]\label{pro: PerronFrob}
Let $M\in \R_+^{n\times n}$. Then, there exists a vector $z\geq 0$, $z\neq 0$ such that $Mz=\rho(M)z$.

Furthermore, for any $\eps>0$ there exists a vector $u>0$ such that $\rho(M)< \|M\|_{\infty}^{u} <\rho(M)+\eps$.
\end{proposition}

\noindent For a proof see Chapter~2, Proposition~6.6 of \cite{BertsekasTsitsiklis1989}.

\subsubsection*{Proposition \ref{pro: weakspectral}}

Now we can characterize the networks that can be rescaled into one with weak externalities.

\begin{proof}[Proof of Proposition \ref{pro: weakspectral}]
Let us start by proving the equivalence of \emph{2}, \emph{3}, and \emph{4}:

\noindent \emph{4}$\Rightarrow$ \emph{2}: Assume that $\rho(|W|-I_n)<1$. Then, by Proposition \ref{pro: PerronFrob} for $0<\eps<1-\rho(|W|-I_n)$, there exists a vector $u\in \R^n$, $u>0$ such that $\| |W|-I_n\|_{\infty}^{u}<1$. Let us use $u^{-1}=a$ as a scaling vector, i.e., let $v_{ij}=w_{ij}a_i/a_j$. Then, we have
\begin{align*}
1>\| |W|-I_n\|_{\infty}^{u}=\| |W|-I_n\|_{\infty}^{a^{-1}}=\| |V|-I_n\|_{\infty}=\| V-I_n\|_{\infty},
\end{align*}
so the matrix $V$ is row diagonally dominant. In other words, $V$ has weak influences.

\noindent \emph{2}$\Rightarrow$ \emph{3}: Let $a\in \R^n$, $a>0$ be a scaling vector so that the rescaled matrix $V$ has weak influences, i.e., $V$ is a row diagonally dominant matrix. Therefore, $\| |V|-I_n\|_{\infty}=\|V-I_n\|_{\infty}<1$, so by (\ref{eq: WVnorm}), we have 
\begin{align*}
\| |W|-I_n\|_{\infty}^{a^{-1}}=\| W-I_n\|_{\infty}^{a^{-1}}= \| V-I_n \|_{\infty}= \| |V|-I_n \|_{\infty}<1.
\end{align*}
For any induced matrix norm $\|\cdot \|$, we have $\| MN\|\leq \|M\|  \|N\|$ for any matrices $M$, $N$. Therefore, $\|M^k\|\leq \|M\|^k$ for any $k\in \N$ and any matrix $M$. Hence,
\[ \lim_{k\to \infty} \| (|W|-I_n)^k \|_{\infty}^{a^{-1}}\leq \lim_{k\to \infty} \left( \| |W|-I_n \|_{\infty}^{a^{-1}}\right)^k=0,\]
since $\| |W|-I_n\|_{\infty}^{a^{-1}}<1$. The norm of $(|W|-I_n)^k$ converges to $0$, this is only possible if the matrices converge to the $0$ matrix, so we have $\lim_{k\to\infty} (|W|-I_n)^k=0$.

\noindent \emph{3}$\Rightarrow$ \emph{4}: Assume that the limit is 0. Let $\lambda$ be any eigenvalue of $|W|-I_n$, and $z\neq 0$ the corresponding eigenvector. Note that $z$ is also an eigenvector of $(|W|-I_n)^k$ with eigenvalue $\lambda^k$. We have
\begin{align*}
0=\left(\lim_{k\to \infty} (|W|-I_n)^k\right) z = \lim_{k\to\infty} (|W|-I_n)^kz=\lim_{k\to\infty} \lambda^kz = \left(\lim_{k\to\infty}\lambda^k\right) z.
\end{align*}
Since $z\neq 0$, we must have $\lim_{k\to\infty} \lambda^k=0$, hence $|\lambda|<1$. This is true for any eigenvalue, so $\rho(|W|-I_n)<1$.

Hence, conditions \emph{2}, \emph{3}, and \emph{4} are equivalent. Now notice that for any matrix $M$, we have $\lim_{k\to\infty} M^k=0$ if and only if $\lim_{k\to \infty} (M^{\top})^k=0$. Therefore, we have the following equivalences: the network $W$ can be rescaled into a row diagonally dominant matrix $\Leftrightarrow$ $\lim_{k\to\infty}(|W|-I_n)^k=0$ $\Leftrightarrow$ $\lim_{k\to\infty}((|W|-I_n)^{\top})^k=0$ $\Leftrightarrow$ $W^{\top}$ can be rescaled into a row diagonally dominant matrix $\Leftrightarrow$ $W$ can be rescaled into a column diagonally dominant matrix.

Column diagonal dominance means exactly that the network has weak externalities, so we proved the equivalence of \emph{1} with the other three statements.
\end{proof}

\subsubsection*{Proposition \ref{pro: weaknash}}

\begin{lemma}\label{le: contraction}
Recall that $b\colon X\to X$ denotes the best response mapping. For any vector $u\in \R^n$, $u>0$ and for all $x,x'\in X$, we have 
\[ \|b(x)-b(x') \|_{\infty}^u \leq \| W-I_n \|_{\infty}^u \|x-x'\|_{\infty}^u.\]
\end{lemma}

\begin{proof}
First, consider the unconstrained best responses $\tilde b(x)$ and $\tilde b(x')$. We have
\begin{align*}
\|\bt(x)-\bt(x') \|_{\infty}^u &=\| (t-(W-I_n)x)- (t-(W-I_n)x')\|_{\infty}^u\\
&= \| (W-I_n)(x'-x) \|_{\infty}^u\\
&\leq \| W-I_n\|_{\infty}^u \|x-x'\|_{\infty}^u \qquad \text{ by (\ref{eq: inducednorm}).}
\end{align*}
Notice that for every $i\in I$, we have $|b_i(x)-b_i(x')|\leq |\bt_i(x)-\bt_i(x')|$. Indeed, without loss of generality, we can assume that $\bt_i(x)\leq \bt_i(x')$ and we can verify the inequality in all cases:
\begin{itemize}[nolistsep]
\item If $\bt_i(x)\leq \bt_i(x')\leq 0$, then $b_i(x)=b_i(x')=0$, so $|b_i(x)-b_i(x')|=0 \leq |\bt_i(x)-\bt_i(x')|$.
\item If $\xh_i\leq \bt_i(x)\leq \bt_i(x')$, then $b_i(x)=b_i(x')=\xh_i$, so $|b_i(x)-b_i(x')|=0 \leq |\bt_i(x)-\bt_i(x')|$.
\item If $\bt_i(x)\leq 0\leq \bt_i(x')$ or if $\bt_i(x)\leq \xh_i \leq \bt_i(x')$, then $\bt_i(x)\leq b_i(x)\leq b_i(x')\leq \bt_i(x')$, and hence $|b_i(x)-b_i(x')|\leq |\bt_i(x)-\bt_i(x')|$.
\item If $0\leq \bt_i(x)\leq \bt_i(x')\leq \xh_i$, then $b_i(x)=\bt_i(x)$ and $b_i(x')=\bt_i(x')$, so $|b_i(x)-b_i(x')|= |\bt_i(x)-\bt_i(x')|$.
\end{itemize}
Therefore, we have
\begin{align*}
\|b_i(x)-b_i(x')\|_{\infty}^u &=\sum_{i\in I} |b_i(x)-b_i(x')|/u_i\\
&\leq \sum_{i\in I} |\bt_i(x)-\bt_i(x')|/u_i\\
&=\|\bt(x)-\bt(x')\|_{\infty}^u\\
&\leq \| W-I_n\|_{\infty}^u \|x-x'\|_{\infty}^u.
\end{align*}
This concludes the proof of the lemma.
\end{proof}

\begin{proof}[Proof of Proposition \ref{pro: weaknash}]
By Proposition \ref{pro: weakspectral}, there exists a scaling vector $a\in\R^n$, $a>0$ such that the rescaled matrix is row diagonally dominant, i.e., we have $\| W-I_n\|_{\infty}^{a^{-1}}=\| V-I_n\|_{\infty}<1$.

Assume that $x^*$ and $x^{**}$ are two different Nash equilibria, i.e., we have $b(x^*)=x^*$, $b(x^{**})=x^{**}$ and $x^*\neq x^{**}$. By the assumption that $\| W-I_n\|_{\infty}^{a^{-1}}<1$ and by Lemma \ref{le: contraction}, we get
\begin{align*}
\| x^*-x^{**} \|_{\infty}^{a^{-1}}=\| b(x^*)-b(x^{**})\|_{\infty}^{a^{-1}}\leq \| W-I_n\|_{\infty}^{a^{-1}} \| x^*-x^{**}\|_{\infty}^{a^{-1}} < \|x^*-x^{**}\|_{\infty}^{a^{-1}}
\end{align*}
which is a contradiction. Hence, there exists a unique Nash equilibrium.
\end{proof}

\subsubsection*{Theorem \ref{thm: weak}}

\begin{proof}[Proof of Theorem \ref{thm: weak}.]
By Proposition \ref{pro: weakspectral}, there exists a scaling vector $a\in\R^n$, $a>0$ such that the rescaled matrix is row diagonally dominant, i.e., we have $\| W-I_n\|_{\infty}^{a^{-1}}=\| V-I_n\|_{\infty}<1$. Let $\gamma=\|W-I_n\|_{\infty}^{a^{-1}}$, then $0\leq \gamma <1$.

By Proposition \ref{pro: weaknash}, there is a unique Nash equilibrium, let us denote it by $x^*$. Notice that a BRD is a BRAD with parameter $0$, so it is enough to prove the statement for BRAD's. Consider any BRAD $(x^k)_{k\in \N}$ with approach parameter $0\leq \beta<1$.

By the definition of the weighted maximum norm, for every $x, y\in X$ we have 
\begin{align}
a_i |x_i -y_i|\leq \|x-y\|_{\infty}^{a^{-1}}. \label{eq: weakthm1}
\end{align}
Hence, for an arbitrary $k\in \N$, we have
\begin{align}
a_{i^k} | b_{i^k}(x^k)-x_{i^k}^*| &= a_{i^k} |b_{i^k}(x^k) - b_{i^k}(x^*) | \nonumber \\
&\leq \| b(x^k)-b(x^*) \|_{\infty}^{a^{-1}} \qquad \text{ by (\ref{eq: weakthm1})} \nonumber \\
&\leq \|W-I_n\|_{\infty}^{a^{-1}} \|x^k-x^*\|_{\infty}^{a^{-1}} \qquad \text{ by Lemma \ref{le: contraction}} \nonumber \\
&=\gamma \cdot \|x^k-x^*\|_{\infty}^{a^{-1}}. \label{eq: weakthm2}
\end{align}

For the next claim, note that since $\gamma +\beta-\gamma\beta = 1-(1-\gamma)(1-\beta)$, we have $0\leq \gamma+\beta-\gamma\beta <1$.

\begin{claim}\label{claim: weak1}
For every $k\in \N$, we have 
\begin{align*}
a_{i^k} |x_{i^k}^{k+1}-x_{i^k}^* |\leq (\gamma+\beta -\gamma\beta) \| x^k - x^* \|_{\infty}^{a^{-1}}.
\end{align*}
\end{claim}

\begin{proof}
Let us use the notation $D=a_{i^k}^{-1}\|x^k-x^*\|_{\infty}^{a^{-1}}\in \R_+$. By (\ref{eq: weakthm1}), we have $|x_{i^k}-x_{i^k}^*|\leq a_{i^k}^{-1}\|x^k-x^*\|_{\infty}^{a^{-1}}=D$, and hence $x_{i^k}^k$ is contained in the interval of lenght $2D$ with midpoint $x_{i^k}^*$. By (\ref{eq: weakthm2}), $b_{i^k}(x^k)$ is in the interval of length $\gamma 2D$ centered at $x_{i^k}^*$.

For $p,q\in \R$, let us use the notation $[p,q]=[\min\{p,q\},\ \max\{p,q\}]$ for the interval between $p$ and $q$.

From the definition of BRAD, we have $x_{i^k}^{k+1}\in [b_{i^k}(x^k),\ (1-\beta)b_{i^k}(x^k)+\beta x_{i^k}^k]$, since this is the $\beta$-contracted image of $[b_{i^k}(x^k),\ x_{i^k}^k]$ towards $b_{i^k}(x^k)$. This implies that $x_{i^k}^{k+1}$ is contained in the $\beta$-contracted image of $[x_{i^k}^*-D,\ x_{i^k}^*+D]$ around $b_{i^k}(x^k)$, which is a point of $[x_{i^k}^*-\gamma D,\ x_{i^k}^*+\gamma D]$. Hence, we get the worst upper bound for $x_{i^k}^{k+1}$ if $b_{i^k}(x^k)$ takes the maximal value in $[x_{i^k}^*-\gamma D,\ x_{i^k}^*+\gamma D]$, and the worst lower bound if $b_{i^k}(x^k)$ takes the minimal value in the interval. We can compute these bounds: if $b_{i^k}(x^k)=x_{i^k}^*+\gamma D$, then the contracted image of $x_{i^k}^*+D$ is 
\[ (1- \beta)(x_{i^k}^* +\gamma D) + \beta(x_{i^k}^* +D)= x_{i^k}^*+ (\gamma+\beta-\gamma\beta)D.\]
Similarly, for $b_{i^k}(x^k)=x_{i^k}^*-\gamma D$ we get the lower bound
\[ (1-\beta)(x_{i^k}^*-\gamma D) + \beta(x_{i^k}^* - D)=  x_{i^k}^*- (\gamma+\beta-\gamma\beta)D.\]
Therefore, $x_{i^k}^{k+1}\in [x_{i^k}^*- (\gamma+\beta-\gamma\beta)D, \ x_{i^k}^*+ (\gamma+\beta-\gamma\beta)D]$, and hence
\[ a_{i^k} |x_{i^k}^{k+1}-x_{i^k}^*| \leq (\gamma+\beta-\gamma\beta)\| x^k-x^*\|_{\infty}^{a^{-1}},\]
as desired.
\end{proof}

\begin{claim}\label{claim: weak2}
For every $m\in \N$ there exists $K(m)\in \N$ such that for all $k\geq K(m)$ we have
\[\|x^k-x^*\|_{\infty}^{a^{-1}}\leq (\gamma+\beta-\gamma\beta)^m \|x^0-x^*\|_{\infty}^{a^{-1}}.\]
\end{claim}

\begin{proof}
We prove the statement by induction on $m$. For $m=0$, it clearly holds with $K(0)=0$. Assume that $K(m-1)$ exists, and we would like to find $K(m)$.

Take an arbitrary $k\geq K(m-1)$ and let player $i$ be the one who moves at time $k$. Then, we have
\begin{align*}
a_i| x_i^{k+1}-x_i^*| &\leq (\gamma+\beta-\gamma\beta) \|x^k-x^*\|_{\infty}^{a^{-1}} \qquad \text{ by Claim \ref{claim: weak1}}\\
&\leq (\gamma+\beta-\gamma\beta) (\gamma+\beta-\gamma\beta)^{m-1} \|x^0-x^*\|_{\infty}^{a^{-1}} \quad \text{ by ind.~hypothesis}\\
&= (\gamma+\beta-\gamma\beta)^m \|x^0-x^*\|_{\infty}^{a^{-1}}.
\end{align*}
Therefore, we can see that $a_i|x_i^{\ell}-x_i^*|\leq \gamma^m \|x^0-x^*\|_{\infty}^{a^{-1}}$ for all $\ell>k$, since it is true after every move of player $i$, and it remains true in all other players' turns because that does not change the action of player $i$.

For every $i\in I$, let $k_i$ be the first time player $i$ moves after $K(m-1)$. Let 
\[K(m)=\max \{k_i : i\in I\} +1.\]
By time $K(m)$, every player moved at least once since $K(m-1)$, so for every $i\in I$ and all $k\geq K(m)$, we have $a_i|x_i^k-x_i^*|\leq (\gamma+\beta-\gamma\beta)^m \|x^0-x^*\|_{\infty}^{a^{-1}}$. Therefore, we also have
\begin{align*}
\|x^k-x^*\|_{\infty}^{a^{-1}} =\max \{a_i|x_i^k-x_i^*| : i\in I\} \leq (\gamma+\beta-\gamma\beta)^m \|x^0-x^*\|_{\infty}^{a^{-1}}.
\end{align*}
This proves the statement for every $m\in \N$.
\end{proof}

Now we can show the convergence of the BRAD $(x^k)_{k\in \N}$ to the Nash equilibrium $x^*$. Take any $\eps>0$, then there exists $m\in \N$ such that $(\gamma+\beta-\gamma\beta)^m \|x^0-x^*\|_{\infty}^{a^{-1}}<\eps$, since $\gamma+\beta-\gamma\beta<1$. Therefore, if $k\geq K(m)$ from Claim \ref{claim: weak2}, then we have 
\begin{align*}
\|x^k-x^*\|_{\infty}^{a^{-1}}\leq (\gamma+\beta-\gamma\beta)^m \|x^0-x^*\|_{\infty}^{a^{-1}}<\eps.
\end{align*}
Hence, $(x^k)_{k\in\N}$ converges to the Nash equilibrium $x^*$.
\end{proof}

\subsubsection*{Proposition \ref{pro: weakpot}}

\begin{lemma}\label{le: rescaleweak}
If the network $W$ can be rescaled into a network with weak externalities using the vector $a\in \mathbb{R}^n$, $a>0$, then for every $j\in I$, we have $\sum_{i\in I\setminus \{j\} } a_i|w_{ij}| < a_j$.
\end{lemma}

\begin{proof}
Take a vector $a\in\mathbb{R}^n$, $a>0$ such that the rescaled network $V$ is with weak externalities. By definition, this means that for every $j\in I$, we have
\begin{align*}
\sum_{i\in I\setminus \{j\} } |v_{ij}| &< |v_{jj}|=1 \\
\sum_{i\in I\setminus \{j\} } \frac{a_i}{a_j}|w_{ij}| &< |w_{jj}|=1 \\
\sum_{i\in I\setminus \{j\} } a_i|w_{ij}| &< a_j.
\end{align*}
\end{proof}

\begin{lemma}
\label{le: effectonbr}
Let $i,j\in I$ and let $(x^k)_{k\in\mathbb{N}}$ be a best-response dynamic. If player $i$'s action changes by $\Delta$, player $j$'s best response changes by a maximum of $| \Delta \cdot w_{ji}|$, i.e., we have $|b_j(x^k)-b_j(x^{k+1})| \leq  |w_{ji}|\cdot | x_i^k-x_i^{k+1}|$ for any $k\in\mathbb{N}$.
\end{lemma}

\begin{proof}
By (\ref{eq: unconstrained}), the unconstrained best response function of player $j$ is $\bt_j(x) = t_j - \sum_{i\in I\setminus \{j\} } w_{ji}x_i$. Therefore, if  $\bt_j(x)\in(0, \xh_j)$, then we have $\bt_j(x) = b_j(x)$, so $\partial b_j(x) / \partial x_i = w_{ji}$. If $\bt_j(x)\notin[0, \xh_j]$, then we know that $b_j(x) \in \{0, \xh_j\}$, so $\partial b_j(x) / \partial x_i = 0$. Hence, by the mean value theorem, if player $j$'s action changes by $\Delta$, her best response can change by at most $|\Delta\cdot w_{ji}|$.
\end{proof}

\begin{proof}[Proof of Proposition \ref{pro: weakpot}.]
Let $i\in I$ and fix $x\in X$. We need to prove that only $b_i(x)$ maximizes $\phi'(\cdot , x_{-i})$.
Assume that we have $x_i^1\in \argmax_{x_i\in X_i} \phi'(x_i,x_{-i})$. Let $x^1=(x_i^1,x_{-i})$ and $i^1=i$. Then $x^2=(b_i(x^1),x_{-i})=(x^2_j)_{j\in I}$.
We have that $\phi'(x^1)\geq \phi'(x^2)$ since $x_i^1\in \argmax_{x_i\in X_i} \phi'(x_i,x_{-i}^1)$. Therefore, we have
\begin{align*}
0&\geq \phi'(x^2)-\phi'(x^1)\\
&= -\left(  \sum_{j\in I\setminus \{i\} } a_j |x^{2}_j-b_j(x^{2})| \right)- \left( - \sum_{j\in I} a_j |x^1_j-b_j(x^1)| \right)\\
&= a_i |x^1_i-b_i(x^1)| +\sum_{j\in I\setminus \{i\} } a_j \left( |x^1_j-b_j(x^1)|-|x^1_j-b_j(x^{2})|\right)\\
&\geq a_i |x^1_i-b_i(x^1)| -\sum_{j\in I\setminus \{i\} } a_j |b_j(x^1)-b_j(x^{2})| \\
&\geq a_i |x^1_i-b_i(x^1)| -\sum_{j\in I\setminus \{i\} } a_j |w_{ji}| |x_i^1-x_i^{2}| \qquad \text{by Lemma \ref{le: effectonbr}}\\
&= \left( a_i -\sum_{j\in I\setminus \{i\} } a_j |w_{ji}| \right) |x_i^1-b_i(x^1)|\\
&\geq 0  \qquad\text{ by Lemma \ref{le: rescaleweak}.}
\end{align*}
Hence, we must have equality everywhere. Since $a_i -\sum_{j\in I\setminus \{i\} } a_j |w_{ji}|>0$, equality holds in the last line if and only if $|x_i^1-b_i(x^1)|=0$, i.e., iff $x_i^1=b_i(x^1)=b_i(x)$. Thus, we have $ \argmax_{x_i\in X_i} \phi'(x_i,x_{-i})= \{b_i(x)\}$, so $\phi'$ is a best-response potential function.
\end{proof}

\subsection{Cycles in large networks with random weights}
In this part we show that if the (ex-ante) possibility of one-way or parasitic interactions exists in an interaction network, even if the probabilities are small, then large networks will have a best-response cycle almost surely. The intuition for this is quite simple: With randomized player interactions directed cycles or parasitic links producing the cycles seen in Examples \ref{ex: circle2} and \ref{ex: parasite} may appear with positive probability. As the number of players goes to infinity, a network allowing for best-response cycles will appear almost surely.

We now turn to a formal description of the above idea. For $i,j\in I$ with $i\neq j$ let the weight $w_{ij}$ be the realization of a random variable denoted by $\tw_{ij}$.

\begin{assumption}\label{ass: ind}
For every $i,i',j,j'\in I$ such that $\{i,j\}\neq \{i',j'\}$ the variables $\tw_{ij}$ and $\tw_{i'j'}$ are independent.
\end{assumption}

\noindent Assumption \ref{ass: ind} states that any two distributions producing weights not belonging to the same pair of players are independent. This assumption is made mostly for convenience, our results may be obtained with less stringent assumptions. Crucially, we allow the two weights describing the interaction between a pair of players to be dependent.

Let $w^-,\uw,\ow>0$ be given such that $\uw<\ow$. We introduce the following notation:
\begin{enumerate}
\item $P_0=\min_{i,j\in I}(P(\tw_{ij}=\tw_{ji}=0))$,
\item $P_1(\uw)=\min_{i,j\in I}(\min\{P(\tw_{ij}\geq\uw,\tw_{ji}=0),P(\tw_{ji}\geq\uw,\tw_{ij}=0)\})$,
\item $P_2(w^-,\uw,\ow)= $ \\
$\min_{i,j\in I}(\min\{P(\tw_{ij}\in[\uw,\ow],\tw_{ji}\leq -w^-), P(\tw_{ji}\in[\uw,\ow],\tw_{ij}\leq -w^-)\})$.
\end{enumerate}

\noindent In words, the smallest probability that two players are indifferent to one another is denoted by $P_0$, the smallest probability of there being a one-directional link of at least strength $\uw$ between an ordered pair is denoted by $P_1(\uw)$, and the smallest probability that an ordered pair has a parasitic link such that the parasite's effect on the host is at least $w^-$ and the host's effect on the parasite falls into the interval $[\uw,\ow]$ is denoted by $P_2(w^-,\uw,\ow)$.

\begin{definition}\label{def: zpn}
We say that the random process generating a network satisfies/allows for
\begin{itemize}
\item \textit{no forced interaction} (NFI) if $P^0>0$,
\item \textit{directed interactions} (DI) if there exists $\uw>0$ such that $P_1(\uw)>0$.
\item \textit{biparasitism} (BP) if there exist $w^-,\uw,\ow>0$ such that $P_2(w^-,\uw,\ow)>0$,
\end{itemize}
\end{definition}

\noindent The properties listed in Definition \ref{def: zpn} are interpreted as follows: \textit{NFI} means that any two players have a positive probability of mutual indifference, meaning that at least a fraction of interaction weights will be zero. This is a feature of the most basic growing random network models where average degree of players is kept constant as $n$ goes to infinity, thus, in fact, most interaction weights are zero. Our assumption is much weaker. \textit{DI} means that every player may unilaterally influence any other with positive probability, a possibility of one-way indifference, while \textit{BP} means that any player may be a parasite of any other with positive probability, an ex-ante possibility of parasitism between any ordered pair. Notice that the three properties impose no restriction on the ex-post realizations of weights. Since all three numbers can be arbitrarily small, the ex-ante restriction is also mild.

\begin{theorem}\label{th: dircycle}
As $n$ goes to infinity, if the random process generating the network satisfies {\normalfont NFI} and {\normalfont DI}, then the resulting directed network game allows for a best-response cycle almost surely.
\end{theorem}

\begin{proof}
Let $\ut=\min_{i\in I}t_i$ and let $\ot=\max_{i\in I}t_i$. For a fixed $w^-,\uw,\ow$, we shorten the notations $P_1(\uw)$ and $P_2(w^-,\uw,\ow)$ to $P_1$ and $P_2$, respectively.

Let $m$ be such that $m\ut\uw\geq\ot$. For some $\ell\in\mathbb{N}$ fix $n$ such that $n=3m\ell$.

We construct three sets of $m$ players, $I_1$, $I_2$, and $I_3$, such that
\begin{itemize}
\item for every pair of players $i,j$, $i\neq j$, belonging to the same group, $w_{ij}=w_{ji}=0$, and
\item for every trio of players $i^1\in I^1$, $i^2\in I^2$, $i^3\in I^3$, it holds that $w_{i^1i^2}=w_{i^2i^3}=w_{i^3i^1}=0$, $w_{i^1i^3},w_{i^2i^1},w_{i^3i^2}>\uw$.
\end{itemize}
\noindent The probability of such realization of weights occurring is at least $P_0^{3m(m-1)/2}P_1^{3m^2}$ which is positive by properties \textit{NFI} and \textit{DI}.

We now show that, given such a group of players, no matter the realization of any other weights in the network, the game has a best-response cycle.

Let $x^1,x^2,x^3\in X$ be given as follows: For $j\in\{1,2,3\}$ $x^j_i=0$ if $i\notin I^j$ and $x^j_i=t_i$ if $i\in I^j$. Let $x^{-1},x^{-2},x^{-3}\in X$ be given as follows: for every $i\in I$ $x^{-1}_i=\max\{x^2_i,x^3_i\}$, $x^{-2}_i=\max\{x^1_i,x^3_i\}$, $x^{-3}_i=\max\{x^1_i,x^2_i\}$.

In profile $x^1$, only the members of group $I^1$ contribute a positive amount, every other player contributes zero. Since $ms\uw\geq\ot$, all members of group $I^2$ are at best response. Since the members of group $I^3$ receive $0$, the profile $x^{-2}$ follows $x^1$ if all members of $I^3$ move to their best responses in any order. Since $ms\uw\geq m\ut\uw\geq\ot$, the best response of the members of group $I^1$ to $x^{-2}$ is to move to $0$. Hence the profile $x^3$ follows $x^{-2}$ by the members of $I^1$ moving to their best responses in any order. Very similarly, $x^{-1}$ follows $x^3$, $x^2$ follows $x^{-1}$, $x^{-3}$ follows $x^2$, and $x^1$ follows $x^{-3}$, completing the cycle.

The probability of three such groups forming in a population of $3m\ell$ players is at least
\begin{equation*}
1-(1-P_0^{3m(m-1)/2}P_1^{3m^2})^\ell.
\end{equation*}
\noindent Let $E_n$ denote the event that a best-response cycle exists in a game of $n$ players. For every $n\geq 0$ it is clear that $P(E_n)\leq P(E_{n+1})\leq 1$, so the sequence $P(E_n)$ is convergent. Therefore, we have
\begin{equation*}
\lim_{n\to\infty}P(E_n)=\lim_{\ell\to\infty}P(E_{3m\ell})\geq 1-\lim_{\ell\to\infty}(1-P_0^{3m(m-1)/2}P_1^{3m^2})^\ell=1.
\end{equation*}
\end{proof}

\begin{theorem}\label{th: parasite}
For $i\in I$ let $\xh_i=\infty$.

As $n$ goes to infinity, if the random process generating the network satisfies {\normalfont NFI} and {\normalfont BP} such that $\ow<\min_{i\in I}t_i/\max_{i\in I}t_i$, then the resulting directed network game allows for a best-response cycle almost surely.
\end{theorem}

\begin{proof}
We introduce the notation $z=\ut-\ow\ot$. Notice that $z>0$ by the assumption of the Theorem. Fix $m\in\mathbb{N}$ such that $\ot-\uw\ut<m\uw z$.

We construct a set of $m$ players, $I_1$, and a player $j\in I\setminus I_1$ such that
\begin{itemize}
\item for every pair of players $i,i'\in I_1$, $i\neq i'$, we have $w_{ii'}=w_{i'i}=0$, and
\item for every $i\in I_1$ it holds that $w_{ij}\in [\uw,\ow]$ and $w_{ji}\leq -w^-$.
\end{itemize}

\noindent The probability of such a realization of weights happening between $m+1$ players is at least $P_0^{m(m-1)/2}P_2^{m}$, which is positive due to properties \textit{NFI} and \textit{BP}.

We now show that, given such a group of players, no matter the realization of any other weights in the network, the game has a best-response cycle. For $i\in I_1$ let $z_i=t_i-w_{i1}t_j$. Notice that $z_i\geq z>0$. Let $x^1, x^2, x^3, x^4$ be given as follows: For every $i\notin I_1\cup\{j\}$ $x^1_i=x^2_i=x^3_i=x^4_i=0$, let $x^1_j=x^2_j=t_j$ and $x^3_j=x^4_j=t_j-w_{ji}z_i$, while for $i\in I_1$ let $x^1_i=x^4_i=0$ and $x^2_i=x^3_i=z_i$.

In profile $x^1$, player $j$ contributes her target, the members of $I^1$ contribute zero. Since $b_i(x^1)=z_i$, a non-negative value by $z_j\geq z>0$, the profile $x^2$ follows from $x^1$ by the members of $I^1$ moving to best response in any order. Since $\xh_j=\infty$, it also holds that $b_j(x^2)=t_j-\sum_{i\in I^1}w_{ji}z_i$, a non-negative value by $w_{ji}\leq -w^- < 0$, hence the profile $x^3$ follows from $x^2$ by $j$ moving to her best response. For $i\in I_1$ we have

\[\hat{b}_i(x^3)=z_i+w_{ij}\sum_{i'\in I^1}w_{ji'}z_{i'}\leq \ot-\uw\ut-m\uw z<0\]
\noindent by the choice of $m$. Hence, for $i\in I^1$ we have $b_i(x^3)=0$ and $x^4$ follows $x^3$ by the members of $I^1$ moving to best response in any order. Finally, $x^1$ follows $x^4$ by player $j$ moving to her best response, completing the cycle.

The probability of such groups forming in a population of $(m+1)\ell$ players is at least
\begin{equation}
1-(1-P_0^{m(m-1)/2}P_2^{m})^\ell.
\end{equation}
\noindent Let $E_n$ denote the event that a best-response cycle exists in a game of $n$ players. For every $n\geq 0$ it is clear that $P(E_n)\leq P(E_{n+1})\leq 1$, hence the sequence $P(E_n)$ is convergent. Then we have
\begin{equation*}
\lim_{n\to\infty}P(E_n)=\lim_{\ell\to\infty}P(E_{(m+1)\ell})\geq 1-\lim_{\ell\to\infty}(1-P_0^{m(m-1)/2}P_2^{m})^\ell=1.
\end{equation*}
\end{proof}

Theorems \ref{th: dircycle} and \ref{th: parasite} establish that, given \textit{NFI}, either of \textit{BP} and \textit{DI} are enough to ensure that large networks will allow for best-response cycles with some additional assumptions on the parameters in the latter case. These properties are mild, in terms of the restrictions they impose on the probability distributions on the weights, as well as intuitive, in terms of what they mean for the links between the players. In large networks, even small likelihood of one-way interactions or parasitic interactions lead to the cycles similar to Examples \ref{ex: circle2} and \ref{ex: parasite} almost surely. Thus we identify directed cycles and parasitism as important stumbling blocks of convergence. Both theorems generalize to non-independent distributions of weights between different pairs of players if the probabilities defined in Definition \ref{def: zpn} are positive, conditional on the possible realizations of other weights.

\newpage
\bibliography{bibliography}
\bibliographystyle{apalike}

\end{document}